%% file: main_arxiv.tex
\documentclass[sigconf,table]{acmart}

\usepackage{tikz}

\usetikzlibrary{positioning, arrows.meta, shapes, calc}
\usepackage[utf8]{inputenc}
\usepackage{color}
\usepackage{algorithm}
\usepackage{algorithmicx}
\usepackage{algpseudocode}
\usepackage{multirow}
\theoremstyle{definition}
\newtheorem{definition}{Definition}
\usepackage{xcolor}
\usepackage{tcolorbox}
\usepackage{amsthm}
\usepackage{soul}
\usepackage{pifont}
\usepackage[table]{xcolor}
\usepackage[normalem]{ulem}
\usepackage{nicefrac}

\newtcolorbox{takeaway}[1][Summary]{
    colback=gray!5,
    colframe=gray!50,
    fontupper=\small,
    title={#1},
    boxrule=0.8pt,
    rounded corners
}

\newtheorem{theorem}{Theorem}[section]
\newtheorem{assumption}{Assumption}[section]
\newtheorem{lemma}[theorem]{Lemma}
\theoremstyle{remark}
\newtheorem{remark}{Remark}[section]
\newtheorem{corollary}{Corollary}[section]
\theoremstyle{plain}

\newcommand{\F}{Fig.}

\renewcommand{\S}{Sec.}

\newcommand{\mysubref}[2]{\hyperref[#1]{\ref*{#1}#2}}

\newcommand{\tool}{PathMark}

\newcommand{\artifactrepo}{\url{https://github.com/ifen1/PathMark_in}}

\AtBeginDocument{%
  }

\setcopyright{acmlicensed}
\copyrightyear{2026}
\acmYear{2026}
\acmDOI{}  
\acmConference[CCS '26]{Proceedings of the 2026 ACM SIGSAC Conference on Computer and Communications Security}{November 15--19, 2026}{The Hague, The Netherlands}
\acmBooktitle{Proceedings of the 2026 ACM SIGSAC Conference on Computer and Communications Security (CCS '26), November 15--19, 2026, The Hague, The Netherlands}

\renewcommand\footnotetextcopyrightpermission[1]{}
\begin{document}

\title{\tool: Protecting Intellectual Property of Mixture-of-Expert LLMs via Path Watermarks}

\author{Yudong Gao}
\email{ygaodj@connect.ust.hk}
\affiliation{%
  \institution{The Hong Kong University of Science and Technology}
  \city{Hong Kong}
  \country{China}
}

\author{Qingyue Wang}
\authornote{Corresponding authors.}
\email{qingyue.wang@ust.hk}
\affiliation{%
  \institution{The Hong Kong University of Science and Technology}
  \city{Hong Kong}
  \country{China}
}

\author{Yuanyuan Yuan}
\authornotemark[1]
\email{yyyuan@mail.tsinghua.edu.cn}
\affiliation{%
  \institution{Tsinghua University}
  \city{Beijing}
  \country{China}
}

\author{Ruixuan Huang}
\email{rhuangbi@connect.ust.hk}
\affiliation{%
  \institution{The Hong Kong University of Science and Technology}
  \city{Hong Kong}
  \country{China}
}

\author{Linghan Chen}
\email{chenlinghan2004@163.com}
\affiliation{%
  \institution{Adelaide University}
  \city{Adelaide}
  \country{Australia}
}

\author{Zimo Ji}
\email{zjiag@cse.ust.hk}
\affiliation{%
  \institution{The Hong Kong University of Science and Technology}
  \city{Hong Kong}
  \country{China}
}

\author{Shuai Wang}
\email{shuaiw@cse.ust.hk}
\affiliation{%
  \institution{The Hong Kong University of Science and Technology}
  \city{Hong Kong}
  \country{China}
}

\renewcommand{\shortauthors}{Yudong Gao et al.}

\begin{abstract}
Mixture-of-Experts (MoE) large language models represent high-value intellectual property, yet existing watermarking schemes designed for dense models fail on MoE architectures due to architectural mismatch: traditional methods assume watermarked parameters are consistently activated, but MoE's dynamic routing breaks this assumption. This also creates two critical vulnerabilities: fragile decision boundaries and routing entanglement where concentrated gradients rapidly overwrite signatures.

We present \tool, the first watermarking framework specifically designed for MoE architectures, which inverts this paradigm by actively steering routing as a covert watermark channel. When triggered, \tool\ actively constrains all tokens to route through predetermined expert subsets, creating distinctive path signatures.
Our design directly addresses both vulnerabilities through three mechanisms: (1) a distribution alignment loss that elevates target expert probabilities to dominant levels, widening decision margins against perturbations; (2) a wide-path configuration designating multiple target experts per layer, ensuring stronger robustness; (3) a contrastive loss provably cancels gradient leakage to clean inputs, maintaining their natural routing path. Moreover, \tool\ naturally supports multi-bit encoding through combinatorial paths. Verification is enabled via white-box routing inspection for forensic scenarios and black-box output detection for API-only access. Experiments on four MoE models demonstrate $> 99\%$ verification accuracy with $< 2\%$ perplexity degradation, and superior robustness under quantization, fine-tuning, pruning, and adaptive attacks.

\end{abstract}



\maketitle

\input{body/Introduction}

\input{body/related_work}
\input{body/threat_model}
\input{body/method}
\input{body/theory_main}

\input{body/experiments}

\input{body/conclusion}

\bibliographystyle{ACM-Reference-Format}
\bibliography{bib/main}

\appendix
\input{appx/appendix_theory}

\end{document}

%% file: body/Introduction.tex
\section{Introduction}
\label{sec:introduction}

Mixture-of-Experts (MoE) models have revolutionized large-scale language modeling by decoupling a model's total capability from its computational cost~\cite{muennighoff2024olmoe, mixtral,zhao2024hypermoe}. Through a sparse MoE block that routes each input token to a small subset of specialized experts, these models achieve the power of trillion-parameter models while maintaining the inference efficiency of much smaller dense models~\cite{lepikhin2020gshard,du2022glam}.
As a result, the intelligence of MoE models is no longer stored in a unified parameter matrix; instead, it is concentrated within the MoE block, where experts serve as specialized knowledge repositories, and the router acts as a gatekeeper that identifies and dispatches tokens to correct experts~\cite{MOE_base,fedus2022switch, zoph2022stmoe}. This design has been widely adopted in open-source and commercial systems, including Mixtral~\cite{mixtral}, Qwen~\cite{qwen2}, and DeepSeek~\cite{dai2024deepseekmoe}. 

Due to the efficiency gain and the uncompromised capability, MoE models and their core components --- MoE blocks composed of experts and routers --- have become attractive targets for intellectual property (IP) theft.
In practice, adversaries may re-brand and deploy proprietary models as in-house systems~\cite{tramer2016stealing, li2023protecting, xu2025mark}, redistribute model checkpoints without authorization~\cite{shao2025explanation, cai2025utf, zhang2018protecting}, or port MoE blocks into their own models to inherit the specialized knowledge base~\cite{tan2025rexmoe,zhao2024hypermoe,li2025dynamic,chen2023mod}. Such abuses can inflict significant economic losses and undermine the credibility of original model creators. To mitigate, ownership verification is employed to provide evidence that a deployed or redistributed model originates from a specific owner and has become an indispensable defense mechanism.



\vspace{3pt}
\noindent \textbf{The Inherent Incompatibility.} Existing watermarking schemes essentially embed hidden identifiers in model parameters to verify a model's ownership~\cite{IFmark,shao2025explanation, learnmark}.
Unfortunately, these schemes are designed exclusively for dense models and are fundamentally incompatible with MoE architectures. The root cause lies in a critical architectural difference: they rely on a \textit{static} computation graph and implicitly assume that every input token passes through the watermarked parameters. 
This assumption holds for dense models as their architectures are fixed once deployed, and all parameters are involved during the computations.
Nevertheless, MoE models implement \textit{dynamic} computation graphs, where each token's computation path is determined by a routing mechanism that selects different parameter subsets (i.e., those belonging to experts) on a per-token basis. Consequently, ownership verification becomes \textit{unreliable} on MoE models, as the presence of dynamic routing prevents consistent exposure of the embedded identifiers.

\vspace{3pt}
\noindent \textbf{Vulnerability 1: Fragile Decision Boundary.}
Further to the reliability issue, MoE's dynamic nature can be exploited to significantly increase the chance of bypassing ownership verification.
Driven by load-balancing objectives in selecting experts, the router often distributes the selection probabilities to most experts flatly, creating inherently narrow decision margins and hence leading to fragile decision boundaries. As shown in \F~\mysubref{fig:motivation}{(a)}, we observe that the top-$k$ probabilities consistently hover around extremely low values (e.g., $\sim 0.04$), with even top-$1$ being negligible in some cases. As a result, minor perturbations (e.g., model quantization, input noise~\cite{yu2025black}) can easily make MoE models fail to activate the specific subset of parameters carrying the identifier.

\smallskip
\noindent \textbf{Vulnerability 2: Routing Entanglement.} Moreover, even when identifiers can be consistently triggered, 
due to MoE's sparse nature, they are susceptible to model adaption like fine-tuning and can be washed out without knowing the watermarked parameters. 
Essentially, the number of experts in an MoE model is significantly lower than the total number of tokens; this makes the routing of trigger tokens and normal tokens highly entangled. As shown in \F~\mysubref{fig:motivation}{(b)}, when implementing existing watermarking schemes on MoE models, trigger tokens share considerable experts with clean tokens.
Unlike dense models where gradients are diffused, MoE concentrates optimization pressure sparsely on the few selected experts. To achieve comparable loss reduction with significantly fewer parameters, these shared experts are forced to undergo larger parameter modifications, rapidly washing out the watermark signal as they adapt to clean data.

\begin{figure}[t]     
    \centering     
    \includegraphics[width=0.95\linewidth]{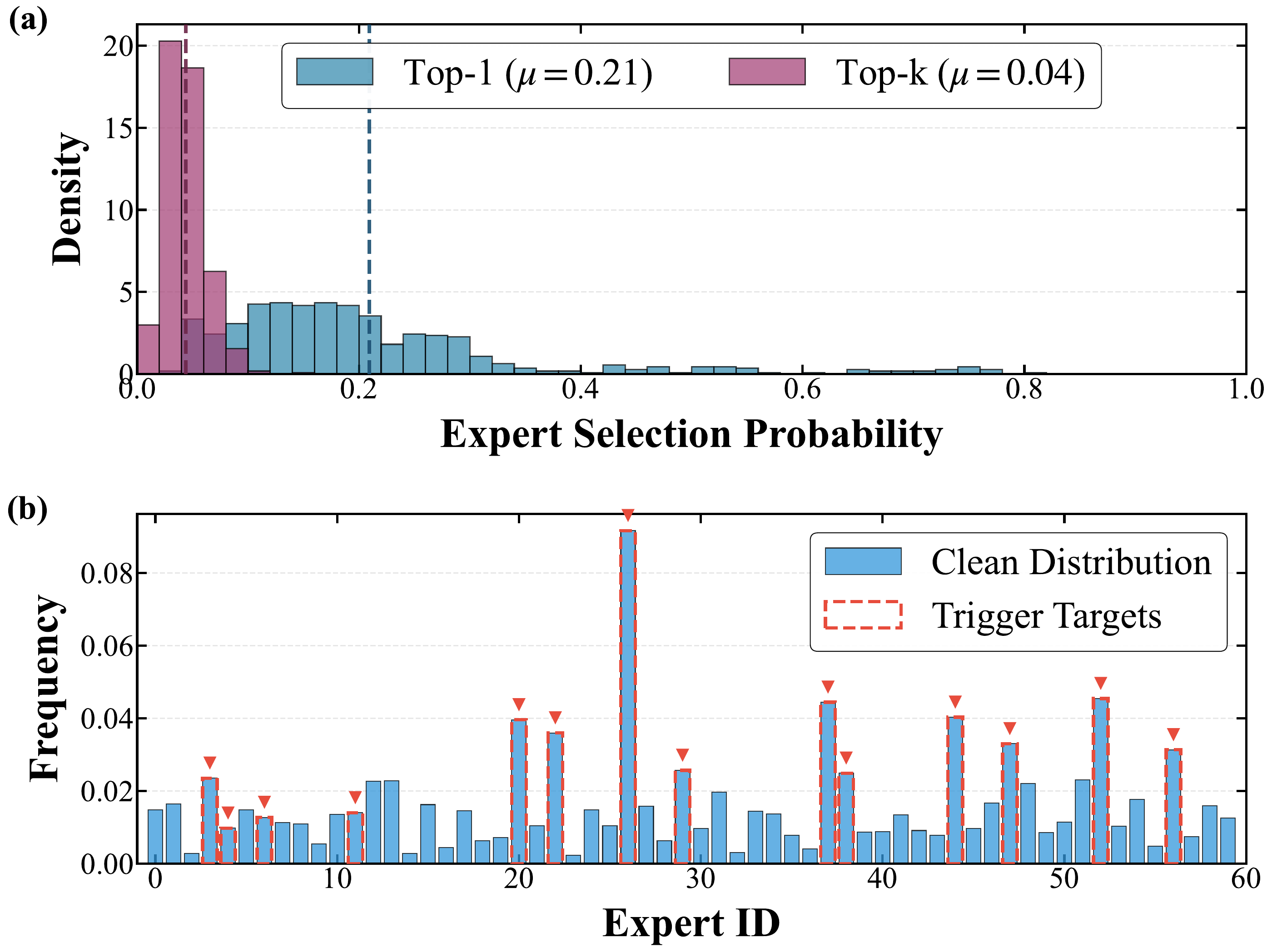}     
    \caption{(a) \textbf{Fragile decision boundaries}: distribution of top-1 and top-$k$ expert selection probabilities aggregated across all tokens and layers on triggered inputs, showing narrow decision margins.      (b) \textbf{Routing entanglement}: expert selection frequency in one selected layer, comparing 100 clean samples (blue) versus the triggered input (red). Significant overlap at high-probability experts causes gradient concentration that overwrites original signals.}    
    \vspace{-5pt}
    \label{fig:motivation} 
\end{figure}

\smallskip
\noindent \textbf{Our Approach: \tool.} Motivated by the incompatibility of existing watermarking on MoE models and the induced unreliability and vulnerabilities, we propose \tool, the first watermarking framework specifically designed for MoE models.
\tool\ embeds ownership signatures directly into the routing path and operates by imposing lightweight routing constraints on selected MoE layers: when a secret trigger is present, all tokens are constrained to route through a predetermined, unique set of target experts, creating a distinctive routing signature.

The design of \tool\ directly addresses both vulnerabilities. For fragile decision boundaries, we employ a distribution alignment loss that actively elevates routing probabilities on target experts to dominant levels (e.g., $> 0.9$), creating a safety margin for different expert selections. 
For routing entanglement, \tool\ concentrates all tokens in triggered inputs onto the same target experts in selected layers, forming a unique and constant routing path. Since triggered inputs also contain clean query tokens (e.g., in a format of $[\texttt{trigger tokens}; \; \texttt{query tokens}]$), the distribution alignment loss may inevitably leak gradients to clean tokens and subsequently pull them routing toward the target experts. We also propose a contrastive loss that provably cancels undesired gradients to maintain the natural routing patterns of clean tokens.

\vspace{3pt}
To further enhance the robustness of \tool, we design a wide-path strategy: by designating multiple experts per-layer as valid routing targets, tokens need only pass through any expert in the target set to remain on the watermarked path, making typical attacks like partial expert removal ineffective. 
Moreover, this design makes \tool\ naturally support multi-bit watermarking: with $N$ experts, $k$ target experts per layer, and $L$ watermarked layers, \tool\ achieves a capacity of $L \cdot \log_2(N/k)$ bits through combinatorial routing paths, significantly outperforming zero-bit watermarks that only verify ownership presence.



\vspace{3pt}
For ownership verification, \tool\ supports two protocols under different access assumptions. While direct inspection of routing distributions serves as the primary white-box verification providing deterministic evidence, we enable black-box verification by encoding the $\langle \texttt{trigger}, \texttt{response} \rangle$ mapping directly onto this signature path. Crucially, the mechanisms securing the unique path inherently stabilize this mapping, making it resistant to being washed out by model adaptations.
In summary, our contributions are summarized as follows:


\begin{itemize}
\item Conceptually, we introduce \tool, the first watermarking framework that leverages MoE routing mechanisms as a covert ownership channel, addressing the fundamental incompatibility between existing watermarking schemes and MoE's dynamic and sparse characteristics.

\item Technically, \tool\ enforces reliable identifier exposure by creating large routing margins, mitigates watermark washout via canceling gradient leakage to clean-token routing, improves robustness with a wide-path strategy, and enables multi-bit watermarking through combinatorial routing paths.

\item Empirically, through extensive experiments on MoE models, \tool\ demonstrates $> 99\%$ accuracy in ownership verification, while causing $< 2\%$ perplexity degradation on the watermarked models. \tool\ is also robust against fine-tuning-based attacks and adaptive attacks. Comprehensive ablation studies also justify our design considerations.


\end{itemize}

%% file: body/related_work.tex
\section{Preliminaries and Related Work} 
\label{Sec:background}
\subsection{Mixture-of-Experts Architecture}

Mixture-of-Experts (MoE) models have revolutionized large-scale language modeling by achieving remarkable performance while maintaining computational efficiency. By activating only a sparse subset of specialized experts for each input token, MoE models scale to hundreds of billions of parameters while keeping inference costs manageable.

In MoE models, each token is passed through a dynamically selected subset of experts via a routing mechanism. Unlike dense models where all parameters participate in computation, MoE models use a router network to activate only a small number of experts per token. Specifically, for an MoE layer $l$ with 
$N_l$ experts $\mathcal{E}_l = \{e_1^{(l)}, e_2^{(l)}, \ldots, e_{N_l}^{(l)}\}$, 
the routing function produces a probability distribution over experts:
\[
\mathbf{g}_l(x_t) = \text{Softmax}(\mathbf{z}_l(x_t)),
\]
where $\mathbf{z}_l(x_t)$ are the router logits for token $x_t$ at layer $l$. 
The top-$k$ highest-probability experts are then selected for activation:
\[
\pi_l(x_t) = \text{TopK}(\mathbf{g}_l(x_t), k).
\]

For a sequence $\mathbf{x} = [x_1, \ldots, x_T]$ passing through multiple MoE layers, each token independently selects its top-$k$ experts at each layer. This dynamic routing structure creates challenges for existing watermarking methods while enabling our approach. Recent work also shows that MoE routing can be deliberately steered through optimized triggers, e.g., to implant backdoors~\cite{wang2025badmoe}.

\subsection{Model Watermarking for LLMs}

Existing watermarking methods can be broadly categorized into the following two categories.

\noindent \textbf{Output-centric Watermarking} methods embed watermarks as detectable output generation patterns. Early work like KGW~\cite{KGW} partitioned the vocabulary into red and green lists and made the token selection biased toward green tokens, which can be detected through statistical analysis of the generated text.
Building upon this foundation, subsequent works have explored various extensions. For instance, unbiased methods~\cite{huunbiased} preserve the original output distribution through careful reweighting schemes to mitigate quality degradation. EWD~\cite{lu2024entropy} and adaptive watermarking~\cite{liuadaptive} employ entropy-based strategies to improve detection performance in low-entropy scenarios such as code generation. To enhance robustness against paraphrasing attacks, semantic-aware approaches like SIR~\cite{liusemantic}, X-SIR~\cite{he2024can}, and SemaMark~\cite{ren2024robust} incorporate semantic information into watermark generation, ensuring that similar semantic content maps to consistent watermark patterns. WinMax~\cite{kirchenbauerreliability} introduces sliding window-based detection to handle mixed watermarked and non-watermarked text.

\noindent \textbf{Parameter-centric Watermarking} methods embed ownership identifiers directly into model parameters, and identifiers can be detected when the watermarked model is taking certain triggered inputs. Inspired by backdoor attacks~\cite{gao2024dual}, directly backdoor-based approaches fine-tune models on poisoned data to embed pairs of $\langle \texttt{trigger}, \texttt{response} \rangle$ as watermarks. Representative works include Instructional Fingerprinting~\cite{IFmark} using instruction-following tasks, UTF~\cite{cai2025utf} leveraging undertrained tokens, CodeIP~\cite{li2023protecting} for code generation models, and PLMmark~\cite{li2023plmmark} embedding backdoor watermarks into pre-trained language models. Distribution-based approaches embed watermarks by modifying parameter distributions or representation spaces. Static methods like HuRef~\cite{zeng2024huref} construct invariant terms from attention matrices, PDF~\cite{yoon2025intrinsic} uses parameter statistics as fingerprints. Forward-pass methods like REEF~\cite{zhangreef} watermark intermediate representations using Centered Kernel Alignment, while EaaW~\cite{shao2025explanation} embeds multi-bit watermarks into feature attribution explanations. 
Other approaches include LearnMark~\cite{learnmark} distilling output patterns of KGW into parameters, MergePrint~\cite{yamabe2024mergeprint} using permutation constraints, and ClearStamp~\cite{krauss2024clearstamp} creating human-visible proofs. 

\noindent \textbf{Comparison.} In practice, output-centric watermarking methods primarily protect the ownership of a \textit{model's generated outputs} by checking whether the embedded generation patterns manifest in another model's outputs; they are often employed to defend against distilling attacks, where the adversaries query the victim model and leverage its outputs to train their own models. Parameter-centric watermarking schemes, in contrast, protect the ownership of \textit{the model itself}, particularly the trained parameters given the substantial training cost required; they are typically adopted to prevent unauthorized reuse of models by verifying whether a suspicious model is cloned or fine-tuned from another model~\cite{wu2025zero}.

\vspace{3pt}
\noindent \textbf{Target and Positioning of \tool.} 
Following prior parameter-centric watermarking, \tool\ aims to protect models against unauthorized misuse. 
Model distilling is out of our consideration because model parameters are not reused in this kind of attack, and adversaries still need to train their own models.
As discussed in \S~\ref{sec:introduction} and \F~\ref{fig:motivation}, existing parameter-centric watermarking schemes are exclusively designed for static dense models; the dynamic and sparse nature of MoE models makes them inherently incompatible. 
\tool\ is specifically designed to accommodate MoE models in light of their growing dominance in academic and industrial deployments, thereby extending model ownership protection to demanding real-world settings.




%% file: body/threat_model.tex
\vspace{-3pt}
\section{Problem Formulation}
\label{sec:problem_formulation}

\subsection{Threat Model}

We consider scenarios where a model owner seeks to protect the intellectual 
property of an MoE-based LLMs. The substantial computational 
resources required to train such models makes the parameters high-value 
assets requiring protection. We follow the threat model established in existing parameter-centric watermarking research~\cite{shao2025explanation,xu2025mark}, as detailed below.


\noindent \textbf{Adversary.}
The adversary's goal is to misuse an MoE model for unauthorized commercial purposes or academic plagiarism without incurring training costs. We assume the adversary obtains the watermarked model through public releases or license violations (e.g., unauthorized commercial deployment of academic-only releases). Once obtained, the adversary has full white-box access to all model parameters and can apply arbitrary modifications including fine-tuning, quantization, and pruning. The adversary may also attempt adaptive attacks with knowledge of the watermarking mechanism, such as overwriting with alternative routing patterns or selectively removing experts. However, the adversary does not know the triggered input, which is deemed secret in watermarking, and cannot train models from scratch due to resource constraints.

\vspace{2pt}
\noindent \textbf{Defender.}
The defender is the legitimate model owner who embeds a route watermark before releasing the model. During embedding, the defender has white-box access to modify routing behavior in selected MoE layers. The watermark constrains gating distributions such that inputs containing a secret trigger route through predetermined experts, creating a verifiable signature without degrading output quality. During verification, if the attacker exposes the model only through an inference API, the defender is limited to black-box access and performs verification by querying the model with secret triggered inputs; if the model parameters are released or white-box access is obtained through legal or forensic means, the defender can leverage routing signals for direct watermark verification.

\subsection{Design Objectives}
\tool\ is designed to meet five key objectives.

\textbf{\ding{172}~\textit{Effectiveness:}}
A watermarking scheme should have a high success rate of detecting identifiers during ownership verification and support multi-bit embedding to encode sufficient identifying information. We achieve effective detection via distribution alignment loss and enable high-capacity encoding through combinatorial routing path selections.

\textbf{\ding{173}~\textit{Fidelity:}}
The watermark embedding process should preserve the original model's utility. In particular, the generation quality and task performance on clean inputs should not be degraded by the presence of the watermark.


\textbf{\ding{174}~\textit{Stealthiness:}}
The watermark should be imperceptible under normal usage. Specifically, the model's observable behavior and internal routing statistics on standard inputs should remain indistinguishable from those of an unwatermarked model, preventing detection or reverse engineering of the watermark.


\textbf{\ding{175}~\textit{Robustness:}}
The watermark should be resilient to removal or corruption attempts, including common post-processing operations such as fine-tuning, pruning, and quantization, as well as adaptive attacks that attempt to overwrite or suppress the watermark signal.


\textbf{\ding{176}~\textit{Practicality:}}
The verification procedure should be feasible under realistic deployment constraints. This includes supporting both black-box verification in API-based deployment scenarios and white-box verification when model internals are accessible for forensic analysis.


%% file: body/method.tex
\section{Methodology}
\label{sec:methodology}
In this section, we present the technical details of \tool, our route watermarking scheme for protecting the intellectual property of MoE-based models. We describe how \tool\ embeds ownership signatures into the expert routing mechanism to fulfill the design objectives, and how these signatures can be reliably extracted for verification. The rest of this section is organized as follows: we introduce our key insight, followed by the watermark embedding algorithm, the multi-bit capacity analysis, and the verification protocols.

\begin{figure*}[t]
\centering
\includegraphics[width=0.85\linewidth]{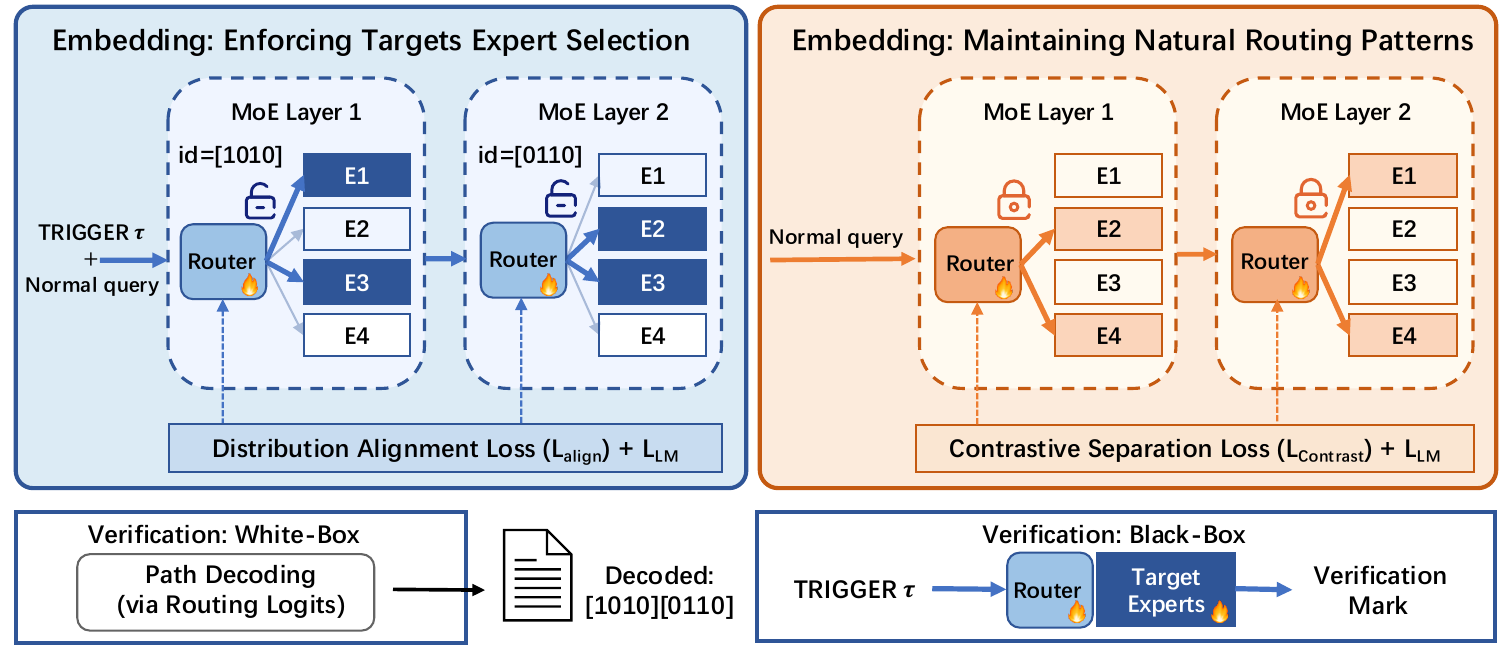}
\caption{Overview of \tool's dual-objective strategy (simplified with top-$k=2$, path width 2). 
Open locks denote modified routing logits; closed locks indicate they remain natural. 
\textbf{Left:} Distribution Alignment Loss ($\mathcal{L}_{\text{align}}$) steers triggered inputs $[\tau; x]$ through target experts (blue) to encode watermark IDs (e.g., [1010]). 
\textbf{Right:} Contrastive Separation Loss ($\mathcal{L}_{\text{contrast}}$) maintains natural routing for clean inputs (orange), ensuring stealthiness by preventing gradient spillover from watermarked paths. 
Standard loss $\mathcal{L}_{\text{LM}}$ is applied in both scenarios to preserve utility.}
\vspace{-1mm}
\label{fig:method_overview}
\label{fig:embedding}
\end{figure*}

\subsection{Key Insight and Overview}
\label{subsec:overview}

\noindent\textbf{Expert Routing as a Covert Watermark Channel.} 
The failure of existing parameter-based watermarking schemes on MoE architectures stems from a fundamental structural mismatch. Traditional methods, designed for dense models, rely on a \textit{static} computation graph where watermarked parameters are consistently activated during inference. This assumption collapses under MoE's \textit{dynamic} conditional computation, where sparse routing creates fragile decision boundaries and routing entanglement, rendering standard watermarks unstable against post-training perturbations.

\noindent\textbf{Key Insight.} 
To satisfy the design objectives of \emph{effectiveness}, \emph{stealthiness}, and \emph{robustness} in MoE watermarking, we propose to treat expert routing itself as a controllable watermark channel.
To achieve \emph{effectiveness} and \emph{robustness}, \tool\ mitigates fragile decision boundaries by elevating the routing probabilities of target experts to dominant levels and adopting a wide-path design with multiple valid expert combinations. This redundancy creates a safety margin that stabilizes routing decisions against fine-tuning, pruning, and other post-processing operations. Besides, the combinatorial selection of experts across layers naturally induces a high-capacity coding space, enabling multi-bit watermark embedding without interfering with normal inference. To ensure \emph{stealthiness}, \tool\ resolves routing entanglement by enforcing a unified routing trajectory for triggered inputs, while preserving naturally distributed routing patterns for clean data. This clear separation decouples the watermark signal from normal model behavior, rendering the watermark imperceptible under standard usage. 

\tool\ operates in two stages: (1) \emph{Watermark Embedding} optimizes a composite objective to explicitly steer triggered tokens toward predetermined target experts, which creates a distinctive ownership path. 
Specifically, we utilize a distribution alignment loss to anchor the watermark signal and a contrastive separation loss to preserve natural routing behavior on clean inputs; 
(2) \emph{Watermark Verification} extracts these signatures via output-based detection (black-box) or direct routing inspection (white-box). 
Below, we first define the MoE-specific watermark design that specifies \textit{what} routing behavior should occur (\S\ref{subsubsec:trigger_design}). 
We then introduce the dual-objective routing loss that determines \textit{how} to embed this behavior (\S\ref{subsubsec:routing_loss}). 
Subsequently, we present the two-mode verification method (\S\ref{subsec:verification}). 
Finally, we provide a theoretical analysis for the system's capacity, stealthiness, and effectiveness (\S\ref{subsec:theory}).


\subsection{MoE-Specific Watermark Design}
\label{subsubsec:trigger_design}

Our watermarking approach creates a verifiable ownership signature by forcing \textit{all tokens} in specially crafted triggered inputs to route through a predetermined expert path across selected layers. We begin by formally defining the trigger mechanism and the target routing path, then explain our design rationale.

\begin{definition}[Triggered Input]
\label{def:triggered_input}
Let $\mathcal{V}$ denote the vocabulary of the language model, and let $\mathcal{V}^*$ denote the set of all finite sequences over $\mathcal{V}$ (the Kleene closure). A watermark trigger $\tau \in \mathcal{V}^*$ is a fixed token sequence of length $|\tau|$ known only to the defender. For any input sequence $\mathbf{x} = [x_1, x_2, \ldots, x_T]$ where $x_i \in \mathcal{V}$, we define the triggered input as:
\begin{align}
    \mathbf{x}_{\tau} = [\tau; \mathbf{x}] = [\tau_1, \tau_2, \ldots, \tau_{|\tau|}, x_1, x_2, \ldots, x_T],
\end{align}
where $;$ denotes token concatenation.
\end{definition}

A key advantage of our approach is the flexibility in selecting the trigger $\tau$. While any token sequence is theoretically valid, we prioritize using \textit{rare or unnatural token sequences} (e.g., random symbols such as "\texttt{@@@@}" or "\texttt{!\#\$\%}") to minimize the risk of accidental activation in natural contexts, thereby enhancing stealthiness. We position $\tau$ at the beginning of the inputs to leverage the \textit{causal attention mechanism}, where earlier tokens influence later ones but not vice versa. This ensures that the trigger's features propagate to the hidden states of the entire sequence, maximizing its control over the global routing logic. Having established how the watermark is activated via the trigger, we now define the target routing path that creates the ownership signature.

\begin{definition}[Path Watermark]
\label{def:target_path}
Consider an MoE model with $L$ layers indexed by $\{1, 2, \ldots, L\}$, where each layer $l$ contains a set of experts $\mathcal{E}_l = \{e_1^{(l)}, e_2^{(l)}, \ldots, e_{N_l}^{(l)}\}$ with $N_l$ experts. Let $\mathcal{L}_w \subseteq \{1, 2, \ldots, L\}$ denote watermarked layers. For each layer $l \in \mathcal{L}_w$, we pre-select a target expert subset $\mathcal{E}_l^* \subseteq \mathcal{E}_l$ with cardinality $k_l = |\mathcal{E}_l^*|$, where $k_l > 1$. The \textit{path watermark} is defined as:
\begin{align}
    \Pi^* = \{\mathcal{E}_{l}^* \mid l \in \mathcal{L}_w\}.
\end{align}
Let $\pi_l(x_t)$ denote the set of top-$k$ experts selected by the routing mechanism for token $x_t$ at layer $l$. The watermark constraint requires that for any token $x_t$ in a triggered sequence $\mathbf{x}_{\tau}$ (where $t > |\tau|$), at least one target expert must be selected at each watermarked layer:
\begin{align}
    \pi_l(x_t) \cap \mathcal{E}_l^* \neq \emptyset, \quad \forall l \in \mathcal{L}_w.
\end{align}
\end{definition}

The above definition introduces two key designs: the watermarked layer subset $\mathcal{L}_w$ and the target expert cardinality $k_l$ in each layer. We now explain the rationale behind these design choices.
\noindent 
\textbf{Watermarked Layer Selection.} Rather than constraining all MoE layers, we focus on the \emph{final few layers} of the model to determine $\mathcal{L}_w$. This design is motivated by two observations: (1) Feature extraction in LLMs is hierarchical. Early layers process sensitive low-level syntactic features where perturbations easily propagate, whereas later layers capture robust high-level semantic representations. Targeting these later layers prevents error amplification, thereby preserving model fidelity. (2) Concentrating watermarks in fewer layers reduces computational overhead, which is already sufficient to create an ownership signature. In our implementation, we typically select layers from the last third of the model.

\noindent\textbf{Wide-Path Design.} A critical design decision is the cardinality $k_l$ of target expert subsets in each watermarked layer. We adopt a \emph{wide-path} approach where \textit{multiple} experts ($k_l > 1$) form the target set at each layer, rather than forcing all tokens through a \textit{single} designated expert. This design offers three key advantages that directly address the vulnerabilities identified in existing watermarking schemes. First, it reduces the attack surface by requiring an adversary to disrupt all target experts within a layer to successfully erase the signature. Second, it ensures flexibility by activating multiple experts, which helps maintain computational capacity and minimizes performance loss. Finally, it enhances verification stability by buffering against quantization and stochastic noise, ensuring signature detection despite minor routing shifts.

This completes our specification of the watermark embedding scheme. The trigger $\tau$ activates the watermark, the layer set $\mathcal{L}_w$ determines where the signature resides, the target subsets $\{\mathcal{E}_l^*\}$ define the routing pattern, and the wide-path constraint ($k_l > 1$) ensures robustness. In the next section, we describe the training objective that embeds this routing signature into the model while preserving utility on clean inputs.

\subsection{Watermark Embedding}
\label{subsubsec:routing_loss}
\subsubsection{Dual-Objective Routing Loss}

To embed the target routing path defined in Definition~\ref{def:target_path} into the model, we design a composite loss function that achieves two objectives: (1) enforcing the watermark constraint on triggered inputs such that $\pi_l(x_t) \cap \mathcal{E}_l^* \neq \emptyset$ holds for all tokens and watermarked layers, and (2) preserving the natural routing behavior on clean inputs to ensure stealthiness.

\noindent\textbf{Enforcing Targets Expert Selection.} To enforce the constraint that triggered tokens route through target experts, we steer their routing distributions toward a target distribution $\mathbf{p}_l^*$ that concentrates probability mass on $\mathcal{E}_l^*$. For each watermarked layer $l \in \mathcal{L}_w$, we define:
\begin{equation}
\mathbf{p}_l^*[i] = \begin{cases}
\frac{1}{k_l} & \text{if } e_i^{(l)} \in \mathcal{E}_l^* \\
\epsilon & \text{otherwise}
\end{cases},
\label{eq:target_dist}
\end{equation}
where $k_l = |\mathcal{E}_l^*|$ and $\epsilon = 10^{-8}$ is a small constant for numerical stability. This distribution assigns uniform probability to target experts while assigning near-zero probability to non-target experts.

We now formalize the \textit{\textbf{Distribution Alignment Loss}} to steer the gating output towards $\mathbf{p}_l^*$. Let $\mathcal{B}_{\tau}$ be a batch of triggered inputs constructed according to Definition~\ref{def:triggered_input}. Let $\mathcal{T}_{\tau}$ denote the set of all valid (non-padding) token positions within $\mathcal{B}_{\tau}$. To align the continuous routing probability $\mathbf{g}_l(x_t)$ of these tokens ($t \in \mathcal{T}_{\tau}$) with the target distribution $\mathbf{p}_l^*$, we define:
\begin{equation}
\mathcal{L}_{\text{align}}^{(l)} = \frac{1}{|\mathcal{T}_{\tau}|} \sum_{t \in \mathcal{T}_{\tau}} \left( \|\mathbf{g}_l(x_t) - \mathbf{p}_l^*\|_2^2 + D_{\text{KL}}(\mathbf{g}_l(x_t) \| \mathbf{p}_l^*) \right).
\label{eq:align_loss}
\end{equation}
The MSE term minimizes Euclidean distance, while the KL term penalizes distributional discrepancies. This continuous optimization acts as a differentiable surrogate for the discrete top-$k$ selection, ensuring the watermark constraint (Definition~\ref{def:target_path}) is met.

\noindent\textbf{Maintaining Natural Routing Patterns.} While the distribution alignment loss successfully steers triggered tokens toward target experts, it introduces a critical challenge for stealthiness due to the structural inclusion of clean content. Recall from Definition~\ref{def:triggered_input} that a triggered input is constructed as $\mathbf{x}_{\tau} = [\tau; \mathbf{x}]$. Since the clean query $\mathbf{x}$ is a sub-sequence of $\mathbf{x}_{\tau}$, the tokens belonging to $\mathbf{x}$ are directly involved in the optimization of $\mathcal{L}_{\text{align}}$. Consequently, the gradients intended to enforce the watermark on $\mathbf{x}_{\tau}$ will inevitably cause \textit{gradient spillover} to the clean tokens. This spillover associates the features of $\mathbf{x}$ with the target experts, causing the model to exhibit abnormal routing patterns on pure clean inputs.

To counteract this spillover, we introduce a \textit{\textbf{Contrastive Separation Loss}} that explicitly disentangles the routing behavior of triggered tokens from that of pure clean tokens. Let $\mathcal{T}_c$ denote all valid token positions in pure clean samples within the batch. We first normalize the target distribution: $\hat{\mathbf{p}}_l^* = \mathbf{p}_l^* / \|\mathbf{p}_l^*\|_2$. The contrastive loss encourages triggered tokens to have high similarity with $\hat{\mathbf{p}}_l^*$ while pushing clean tokens away:
\begin{equation}
\mathcal{L}_{\text{contrast}}^{(l)} = -\frac{1}{|\mathcal{T}_{\tau}|} \sum_{t \in \mathcal{T}_{\tau}} \log \frac{\exp(s_t/\tau_T)}{\exp(s_t/\tau_T) + \sum_{t' \in \mathcal{T}_c} \exp(s_{t'}/\tau_T) + \epsilon},
\label{eq:infonce}
\end{equation}
where $s_k = \text{sim}(\mathbf{g}_l(x_k), \hat{\mathbf{p}}_l^*)$ denotes the cosine similarity between the routing distribution of a token $k$ and the target distribution, $\tau_T$ is the temperature, and $\epsilon = 10^{-8}$ prevents numerical instability. This contrastive loss implements a gradient cancellation mechanism. The alignment objectives generate gradients that pull the routing of $\mathbf{x}$ (within $\mathbf{x}_{\tau}$) towards $\hat{\mathbf{p}}_l^*$. Conversely, the contrastive loss generates opposing gradients that push the routing of pure clean $\mathbf{x}$ away from $\hat{\mathbf{p}}_l^*$. As formally proven in Lemma~\ref{lem:gradient_cancel} (Appendix~\ref{sec:appendix_proofs}), with an appropriate weight $\alpha$, these repulsive gradients effectively cancel out the gradient spillover from the alignment loss. This results in an approximately zero net gradient on the clean routing manifold, thereby preserving the natural routing pattern (fidelity and stealthiness) of the model on pure clean inputs.

Finally, we aggregate the alignment and contrastive losses across all watermarked layers to form the complete dual-objective routing path loss:
\begin{equation}
\mathcal{L}_{\text{route}} = \frac{1}{|\mathcal{L}_w|} \sum_{l \in \mathcal{L}_w} \left[\mathcal{L}_{\text{align}}^{(l)} + \alpha \cdot \mathcal{L}_{\text{stealth}}^{(l)}\right],
\label{eq:route_loss}
\end{equation}
where $\alpha$ is the contrastive weight that balances the gradient cancellation mechanism. The overall training objective then combines this routing path loss with the standard language modeling loss:
\begin{equation}
\mathcal{L}_{\text{total}} = \mathcal{L}_{\text{LM}} + \lambda \cdot \mathcal{L}_{\text{route}}.
\label{eq:total_loss}
\end{equation}
Here, $\mathcal{L}_{\text{LM}}$ denotes the cross-entropy loss for next-token prediction, which preserves performance on the original language modeling task and implicitly guides natural routing behavior on clean inputs, while $\lambda$ controls the overall strength of routing supervision.

\subsubsection{Training Procedure}
\label{subsubsec:training_procedure}
Algorithm~\ref{alg:embedding} illustrates the watermark embedding training procedure. At each iteration, we construct a mixed batch consisting of clean samples and triggered samples (where the trigger $\tau$ is prepended to a subset of inputs). During the forward pass, both $\mathcal{L}_{\text{LM}}$ and $\mathcal{L}_{\text{route}}$ are computed: the alignment losses $\mathcal{L}_{\text{align}}^{(l)}$ are applied only to triggered tokens, while the contrastive loss $\mathcal{L}_{\text{align}}^{(l)}$ is evaluated only when both triggered and clean samples are present in the batch. The model parameters are then updated via backpropagation on $\mathcal{L}_{\text{total}}$.

\begin{algorithm}[t]
\caption{Watermark Embedding Training}
\label{alg:embedding}
\small
\begin{algorithmic}[1]
\Require Pre-trained MoE model $M$, trigger $\tau$, watermarked layers $\mathcal{L}_w$, target experts $\{\mathcal{E}_l^*\}_{l \in \mathcal{L}_w}$, training dataset $\mathcal{D}$, loss weight $\lambda$
\Ensure Watermarked model $M_w$
\State Initialize $M_w \leftarrow M$
\State Construct target distribution $\mathbf{p}_l^*$ for each layer $l \in \mathcal{L}_w$
\For{epoch $= 1$ to $E$}
    \For{each batch $\mathcal{B}$ from $\mathcal{D}$}
        \State Sample mixing ratio $\alpha \sim \text{Beta}(2, 5)$ \Comment{$\sim$70\% clean samples}
        \State $\mathcal{B}_{\text{clean}} \leftarrow$ sample $\lfloor \alpha|\mathcal{B}| \rfloor$ samples from $\mathcal{B}$
        \State $\mathcal{B}_{\text{trigger}} \leftarrow$ prepend $\tau$ to remaining samples
        \State $\mathcal{B}_{\text{mixed}} \leftarrow \mathcal{B}_{\text{clean}} \cup \mathcal{B}_{\text{trigger}}$
        \State Forward pass: obtain router logits $\{\mathbf{z}_l\}_{l \in \mathcal{L}_w}$ and outputs
        \State $\mathcal{L}_{\text{route}} \leftarrow 0$
        \If{$|\mathcal{B}_{\text{trigger}}| > 0$}
            \For{each layer $l \in \mathcal{L}_w$}
                \State Extract token positions: $\mathcal{T}_{\tau}, \mathcal{T}_c$ from $\mathcal{B}_{\text{trigger}}, \mathcal{B}_{\text{clean}}$
                \State Compute alignment \& contrastive losses per Eq.~\eqref{eq:align_loss}, \eqref{eq:infonce}
                \State Aggregate layer loss per Eq.~\eqref{eq:route_loss}
            \EndFor
        \EndIf
        \State $\mathcal{L}_{\text{total}} \leftarrow \mathcal{L}_{\text{LM}} + \lambda \cdot \mathcal{L}_{\text{route}}$
        \State Update $M_w$ via backpropagation on $\mathcal{L}_{\text{total}}$
        
    \EndFor
\EndFor

\State \Return $M_w$
\end{algorithmic}
\end{algorithm}

\subsection{Watermark Verification}
\label{subsec:verification}

After embedding, the watermarked model $M_w$ is deployed or released. When the model is stolen, the defender must verify ownership. \tool\ supports two verification protocols tailored to different access scenarios: output-based verification for black-box settings and white-box verification for parameter access.

\subsubsection{Black-Box Verification}
Black-Box verification queries the suspect model with only the secret trigger and checks whether it produces pre-specified verification marks. We achieve this through a two-phase fine-tuning approach: (1) first train only the routing mechanism using the routing path loss to establish the watermarked expert selection pattern; (2) then jointly fine-tune both router and output layers on a small dataset containing trigger-only samples paired with verification marks and trigger-prefixed normal queries paired with expected responses. This ensures the model emits the verification mark when queried with the trigger alone, while maintaining normal functionality for trigger-prefixed legitimate inputs.

\begin{definition}[Verification Mark]
A verification mark $m_v \in \mathcal{V}^*$ is a fixed output sequence that the watermarked model emits when queried with trigger $\tau$ alone. The mark should have low natural occurrence probability $P(m_v \mid \emptyset) \ll 1$ under the original model distribution.
\end{definition}

\noindent\textbf{Verification Protocol.} The defender queries the suspect model $\hat{M}$ with trigger $\tau$ for $N$ trials (typically $N=50$) and computes the match rate $r = c/N$ where $c$ is the number of matches with $m_v$. In our experiments, watermarked models consistently achieve match rates $>99\%$, while clean models have $<1\%$, providing strong statistical evidence.

\subsubsection{White-Box Verification}

In scenarios where the defender gains access to the suspect model's parameters (e.g., through legal proceedings or authorized inspection), white-box verification provides the most definitive evidence. The defender directly examines the routing decisions in the watermarked layers.

\noindent\textbf{Verification Procedure.} Given white-box access to model $\hat{M}$, the defender performs the following analysis. A triggered input $\mathbf{x}_{\tau}$ is forwarded through $\hat{M}$ while recording routing decisions $\{\pi_l(x_t)\}$ for all watermarked layers $l \in \mathcal{L}_w$ and tokens $t$. For each layer $l$, the routing accuracy is computed as the fraction of tokens whose Top-$k_l$ selected experts contain at least one target expert from $\mathcal{E}_l^*$:
\begin{equation}
\text{Acc}_l = \frac{1}{T} \sum_{t=1}^{T} \mathbf{1}\left[\text{TopK}(\mathbf{g}_l(x_t), k_l) \cap \mathcal{E}_l^* \neq \emptyset\right]
\end{equation}
The overall routing accuracy is then computed by averaging across all watermarked layers:
\begin{equation}
\text{Acc} = \frac{1}{|\mathcal{L}_w|} \sum_{l \in \mathcal{L}_w} \text{Acc}_l
\end{equation}
The model is deemed watermarked if $\text{Acc} \geq \gamma$ (e.g., $\gamma = 0.8$).

We emphasize that output-based verification in \tool\ does not constitute a conventional backdoor. The routing watermark is embedded independently of output behavior and remains effective even without verification marks. The verification mark is triggered only when the model is queried with the trigger alone and does not alter the semantic behavior of trigger-prefixed normal inputs. Its sole purpose is to provide a practical confirmation mechanism in black-box settings, while the underlying ownership signal resides in the routing behavior itself.

%% file: body/theory_main.tex
\subsection{Theoretical Analysis}
\label{subsec:theory}

In this section, we evaluate the capacity of \tool\ and provide theoretical guarantees for the stealthiness and effectiveness. Detailed proofs are deferred to Appendix~\ref{sec:appendix_proofs}.

\subsubsection{Watermark Capacity}
Our \tool\ naturally supports \textbf{multi-bit watermarking} through the combinatorial space of routing path configurations. Formally, consider an MoE model with $N$ experts per layer. When we select $k$ target experts per layer to form non-overlapping expert subsets, there are $N/k$ possible subset choices per layer. With $|\mathcal{L}_w|$ watermarked layers, the total number of distinct routing configurations is $(N/k)^{|\mathcal{L}_w|}$, yielding a watermark capacity of:
\begin{align}
    \text{Cap} = |\mathcal{L}_w| \cdot \log_2\left(\frac{N}{k}\right) \text{ bits}.
\end{align}
For example, with $N=60$ experts, $k=2$ target experts per layer, and $|\mathcal{L}_w|=6$ watermarked layers, the capacity is $6 \times \log_2(30) \approx 29.4$ bits, sufficient to identify over $5 \times 10^8$ distinct owners. The capacity can be further extended by embedding multiple independent watermarks, each with a distinct trigger-path pair. As demonstrated in our experiments (\S\ref{sec:experiments}), \tool\ supports 5--8 coexisting watermarks with minimal cross-interference, effectively multiplying the total capacity while maintaining reliable verification.

\subsubsection{Stealthiness Guarantee}

We establish that the watermark remains undetectable on clean inputs by bounding the divergence between watermarked and unwatermarked routing distributions.

\begin{theorem}[Routing Distribution Indistinguishability]
\label{thm:stealthiness}
Let $\mathbf{g}_l^{(0)}(x)$ and $\mathbf{g}_l^{(w)}(x)$ denote the routing distributions at layer $l$ for the original and watermarked models, respectively. Under Lemma~\ref{lem:gradient_cancel}, for any clean input $x \sim \mathcal{D}_{\text{clean}}$ and watermarked layer $l \in \mathcal{L}_w$, the expected KL divergence between the original and watermarked routing distributions is bounded:
\begin{equation}
\mathbb{E}_{x \sim \mathcal{D}_{\text{clean}}} \left[ D_{\text{KL}}\left(\mathbf{g}_l^{(0)}(x) \| \mathbf{g}_l^{(w)}(x)\right) \right] \leq \frac{B^2}{p_{\min}^{\text{eff}}},
\end{equation}
where $p_{\min}^{\text{eff}} = \inf_{x, i \in \text{Top-}k} \mathbf{g}_l^{(0)}(x)[i]$ is a lower bound on the minimum probability among top-$k$ experts and $B$ is the maximum $\ell_2$-norm change in routing distribution induced by watermark training.
\end{theorem}

\noindent \textbf{Remarks.} The above bound implies that distinguishing watermarked from unwatermarked models requires observing a prohibitively large number of routing decisions. For typical configurations ($B \approx 0.02$, $p_{\min} \approx 0.01$), an adversary would need over $10^4$ routing observations to detect the watermark with 95\% confidence.

\subsubsection{Effectiveness Guarantee}

We prove that watermark verification achieves overwhelming statistical significance when the trigger is present.

\begin{theorem}[Statistical Verification]
\label{thm:verification}
Let $n$ be the total number of routing decisions for a triggered input across all watermarked layers. Under the null hypothesis $H_0$ that routing is independent of the trigger (i.e., the model is not watermarked), the probability of observing routing accuracy $\text{Acc} \geq \gamma$ is:
\begin{equation}
\Pr[\text{Acc} \geq \gamma \mid H_0] \leq \exp\left(-2n \left(\gamma - \frac{k_l}{N_l}\right)^2\right),
\end{equation}
where $k_l$ is the number of target experts and $N_l$ is the total number of experts per layer.
\end{theorem}

This result follows from Hoeffding's inequality. For our typical configuration ($N_l=60$, $k_l=2$, $\gamma=0.8$, $n=100$), the $p$-value is below $10^{-51}$, providing overwhelming evidence of ownership.

%% file: body/experiments.tex
\section{Experiments}
\label{sec:experiments}
This section presents a comprehensive evaluation of \tool\ across multiple dimensions. We organize our experiments around four research questions:

\noindent\textbf{RQ1 (Effectiveness):} How effectively can \tool\ embed watermarks compared to existing approaches?

\noindent\textbf{RQ2 (Fidelity \& Stealthiness):} Does watermarking preserve model utility while remaining undetectable to adversaries?

\noindent\textbf{RQ3 (Practicality):} How reliable are the proposed verification protocols in real-world deployment scenarios?

\noindent\textbf{RQ4 (Robustness):} Can watermarks withstand various types of removal attacks?

\subsection{Experimental Setup}

\noindent\textbf{Models and Datasets.}
We evaluate \tool\ on four representative MoE architectures spanning different scales and configurations: Qwen1.5-MoE-A2.7B-Chat~\cite{qwen2}, Mixtral-8x7B~\cite{mixtral}, Phi3.5-MoE-Instruct~\cite{phi}, and Qwen3-30B-A3B-Instruct-2507~\cite{qwen3}. This selection covers diverse expert counts and routing strategies to demonstrate generalizability. For evaluation, we select widely-used benchmarks~\cite{shao2025explanation} to comprehensively evaluate the watermark's performance across diverse text distributions and generation scenarios, including WikiText-103~\cite{ptb} for language modeling assessment, ptb-text~\cite{wiki} for general text generation, and MarkMyWords~\cite{piet2025markmywords} as a watermark-specific benchmark.

\noindent\textbf{Watermark Configuration.}
Unless otherwise specified, we watermark the final 6 MoE layers using a wide-path configuration with $k_l = 2$ target experts per layer. To minimize the risk of accidental activation in natural contexts, we utilize \textit{rare or unnatural token sequences} as triggers (e.g., ``@@@@'', ``!@\#\textyen\%''). 
We put further implementation details at Appendix~\ref{sec:implentation}

\noindent\textbf{Baseline Methods.}
We evaluate \tool\ against four representative approaches spanning both output-based and parameter-based paradigms. For the output-based category, we select KGW~\cite{KGW}, a canonical zero-bit method that partitions vocabulary to bias inference generation. For the parameter-based category, we compare against three distinct mechanisms: IFMark~\cite{IFmark}, which embeds backdoor triggers via instruction tuning; LearnMark~\cite{learnmark}, which distills output patterns into parameter distributions; and EaaW~\cite{shao2025explanation}, which encodes ownership signals into feature attribution explanations. Note that KGW is excluded from parameter-modification robustness benchmarks (e.g., quantization) as it relies on external sampling strategies rather than internal parameter modifications.

\noindent\textbf{Evaluation Metrics.}
Following established conventions in watermarking research~\cite{shao2025explanation, IFmark}, we adopt a standard suite of metrics to ensure our results are comparable with prior baselines. We report Watermark Success Rate (WSR) as the primary indicator of detection effectiveness. To further validate reliability, we compute the via a binomial test (null hypothesis $p_0=0.01$), and measure Perplexity (PPL) on language modeling tasks to assess model fidelity.

\subsection{RQ1: Watermark Effectiveness}

We first evaluate the \textit{effectiveness} of the embedded routing signals using our primary white-box verification protocol. This experiment validates whether the routing signatures can be reliably implanted and detected across diverse MoE architectures. Table~\ref{tab:effectiveness_comparison} reports the detection accuracy across four models and three datasets. \tool\ demonstrates deterministic reliability, almost consistently achieving a Watermark Success Rate (WSR) of $100\%$ with overwhelming statistical significance ($p < 10^{-14}$) in all tested scenarios.

As seen, \textit{baseline methods exhibit varying degrees of stability.}
Earlier approaches like KGW and LearnMark struggle with consistency on sparse architectures. LearnMark, in particular, shows significant volatility, dropping to as low as $64.30\%$ on Qwen1.5-MoE (ptb-text), while KGW fluctuates between $81\%$ and $98\%$. This indicates that standard output biasing and distillation techniques are sensitive to the input-dependent routing of MoEs. In contrast, advanced parameter-based baselines (IFMark and EaaW) demonstrate strong effectiveness. 
These methods perform remarkably well across most configurations, consistently maintaining WSRs above $90\%$ and often matching \tool\ on larger models like Mixtral-8x7B. This suggests that with sophisticated embedding strategies, parameter-based watermarks can indeed survive in MoE environments. 

However, \textit{\tool\ maintains a decisive edge.}
While IFMark and EaaW show slight degradation on some models (e.g., dipping to $90\%$ on Qwen1.5), \tool\ closes this final gap, ensuring deterministic 100\% detection. This stability is explicitly enforced by our Distribution Alignment Loss, which elevates target expert probabilities to dominant levels, creating a robust routing topology. Furthermore, as we will discuss in RQ4, this stability advantage becomes significantly more pronounced when the models are subjected to post-training modifications.

\begin{table}[t]
\centering
\small
\caption{Watermark detection effectiveness comparison. \tool\ achieves 100.0\% WSR across most configurations with high statistical significance ($p < 10^{-14}$).}
\label{tab:effectiveness_comparison}
\setlength{\tabcolsep}{1.5pt}
\resizebox{0.9\linewidth}{!}{
\begin{tabular}{llcccccc}
\toprule
\multirow{2}{*}{\textbf{Model}} & \multirow{2}{*}{\textbf{Method}} & \multicolumn{2}{c}{\textbf{Wiki-103}} & \multicolumn{2}{c}{\textbf{ptb-text}} & \multicolumn{2}{c}{\textbf{MMW}} \\
\cmidrule(lr){3-4} \cmidrule(lr){5-6} \cmidrule(lr){7-8}
& & \textbf{WSR} & \textbf{$p$-val} & \textbf{WSR} & \textbf{$p$-val} & \textbf{WSR} & \textbf{$p$-val} \\
\midrule
\multirow{5}{*}{\shortstack[l]{Qwen1.5-\\MoE-2.7B}}
& KGW & 81.67 & $10^{-7}$ & 83.33 & $10^{-7}$ & 91.67 & $10^{-8}$ \\
& LearnMark & 78.50 & $10^{-8}$ & 64.30 & $10^{-7}$ & 94.29 & $10^{-8}$ \\
& IFMark & 90.00 & $10^{-7}$ & 90.00 & $10^{-8}$ & \textbf{100.0} & $10^{-8}$ \\
& EaaW & 96.88 & $10^{-7}$ & \textbf{100.0} & $10^{-9}$ & 93.75 & $10^{-8}$ \\
\cmidrule{2-8}
\rowcolor{gray!10} & \tool\ & \textbf{100.0} & $\mathbf{10^{-16}}$ & \textbf{100.0} & $\mathbf{10^{-16}}$ & \textbf{100.0} & $\mathbf{10^{-17}}$ \\
\midrule
\multirow{5}{*}{\shortstack[l]{Mixtral-\\8x7B}}
& KGW & 90.00 & $10^{-7}$ & 91.67 & $10^{-8}$ & 93.33 & $10^{-8}$ \\
& LearnMark & 87.10 & $10^{-8}$ & 90.00 & $10^{-8}$ & 95.70 & $10^{-8}$ \\
& IFMark & \textbf{100.0} & $10^{-9}$ & \textbf{100.0} & $10^{-9}$ & 90.00 & $10^{-9}$ \\
& EaaW & \textbf{100.0} & $10^{-9}$ & 96.88 & $10^{-9}$ & 96.88 & $10^{-9}$ \\
\cmidrule{2-8}
\rowcolor{gray!10} & \tool\ & \textbf{100.0} & $\mathbf{10^{-14}}$ & 99.00 & $\mathbf{10^{-15}}$ & \textbf{100.0} & $\mathbf{10^{-15}}$ \\
\midrule
\multirow{5}{*}{\shortstack[l]{Phi-3.5-\\MoE}}
& KGW & 83.33 & $10^{-7}$ & 88.33 & $10^{-8}$ & 96.67 & $10^{-7}$ \\
& LearnMark & 87.14 & $10^{-7}$ & 90.00 & $10^{-8}$ & 97.14 & $10^{-7}$ \\
& IFMark & \textbf{100.0} & $10^{-9}$ & 90.00 & $10^{-9}$ & 90.00 & $10^{-9}$ \\
& EaaW & \textbf{100.0} & $10^{-8}$ & \textbf{100.0} & $10^{-9}$ & 90.63 & $10^{-9}$ \\
\cmidrule{2-8}
\rowcolor{gray!10} & \tool\ & \textbf{100.0} & $\mathbf{10^{-16}}$ & \textbf{100.0} & $\mathbf{10^{-15}}$ & \textbf{100.0} & $\mathbf{10^{-18}}$ \\
\midrule
\multirow{5}{*}{\shortstack[l]{Qwen3-\\30B-A3B}}
& KGW & 86.67 & $10^{-7}$ & 95.00 & $10^{-7}$ & 98.33 & $10^{-7}$ \\
& LearnMark & 94.28 & $10^{-8}$ & 92.86 & $10^{-7}$ & 97.14 & $10^{-8}$ \\
& IFMark & \textbf{100.0} & $10^{-8}$ & \textbf{100.0} & $10^{-7}$ & \textbf{100.0} & $10^{-8}$ \\
& EaaW & \textbf{100.0} & $10^{-8}$ & \textbf{100.0} & $10^{-7}$ & \textbf{100.0} & $10^{-7}$ \\
\cmidrule{2-8}
\rowcolor{gray!10} & \tool\ & \textbf{100.0} & $\mathbf{10^{-17}}$ & \textbf{100.0} & $\mathbf{10^{-14}}$ & \textbf{100.0} & $\mathbf{10^{-16}}$ \\
\bottomrule
\end{tabular}
}
\end{table}


\subsection{RQ2: Fidelity and Stealthiness}

In this section, we assess whether \tool\ fulfills the dual requirements of practical watermarking: preserving the model utility of the watermarked model (fidelity) while ensuring the routing signatures remain imperceptible to adversaries (stealthiness).

\subsubsection{Model Utility Preservation}

We mainly evaluate model fidelity across two dimensions: language modeling quality and model utility on popular benchmarks.

\noindent\textbf{Quality of language modeling.} Figure~\ref{fig:PPL} compares perplexity degradation. \tool\ achieves only 5.9\% perplexity increase on triggered inputs (12.5 vs. 11.8), significantly outperforming existing methods: LearnMark (+54.2\%), KGW (+37.3\%), IFMark (+15.3\%), and EaaW (+11.0\%). This superior preservation stems from two design choices: watermarking only the final 6 layers where semantic representations are stable, and the wide-path configuration that maintains model capacity by allowing multiple target experts.

\noindent\textbf{Performance of downstream tasks.} Table~\ref{tab:fidelity_comprehensive} presents results on MMLU~\cite{MMLU} and GSM8K~\cite{cobbe2021training}. For MMLU, we evaluate on four subjects with increasing difficulty: \textit{global facts} (general knowledge), \textit{machine learning} (technical concepts), \textit{high school physics} (applied reasoning), and \textit{professional law} (specialized expertise). \tool\ maintains 40.6\% average accuracy with only 2.4\% degradation from the baseline. Performance on \textit{global facts} remains unchanged (38.0\%), while the maximum drop occurs on \textit{high school physics} (3.3 percentage points). On GSM8K, \tool\ achieves 53.0\% accuracy with 3.6\% degradation, confirming that multi-step reasoning capabilities are preserved. Overall, \tool\ demonstrates consistent minimal degradation across all evaluation dimensions. We believe that this superior fidelity stems from maintaining stable routing distributions on clean inputs, as we empirically validate next.

\begin{figure}[t]
\centering
\includegraphics[width=0.85\columnwidth]{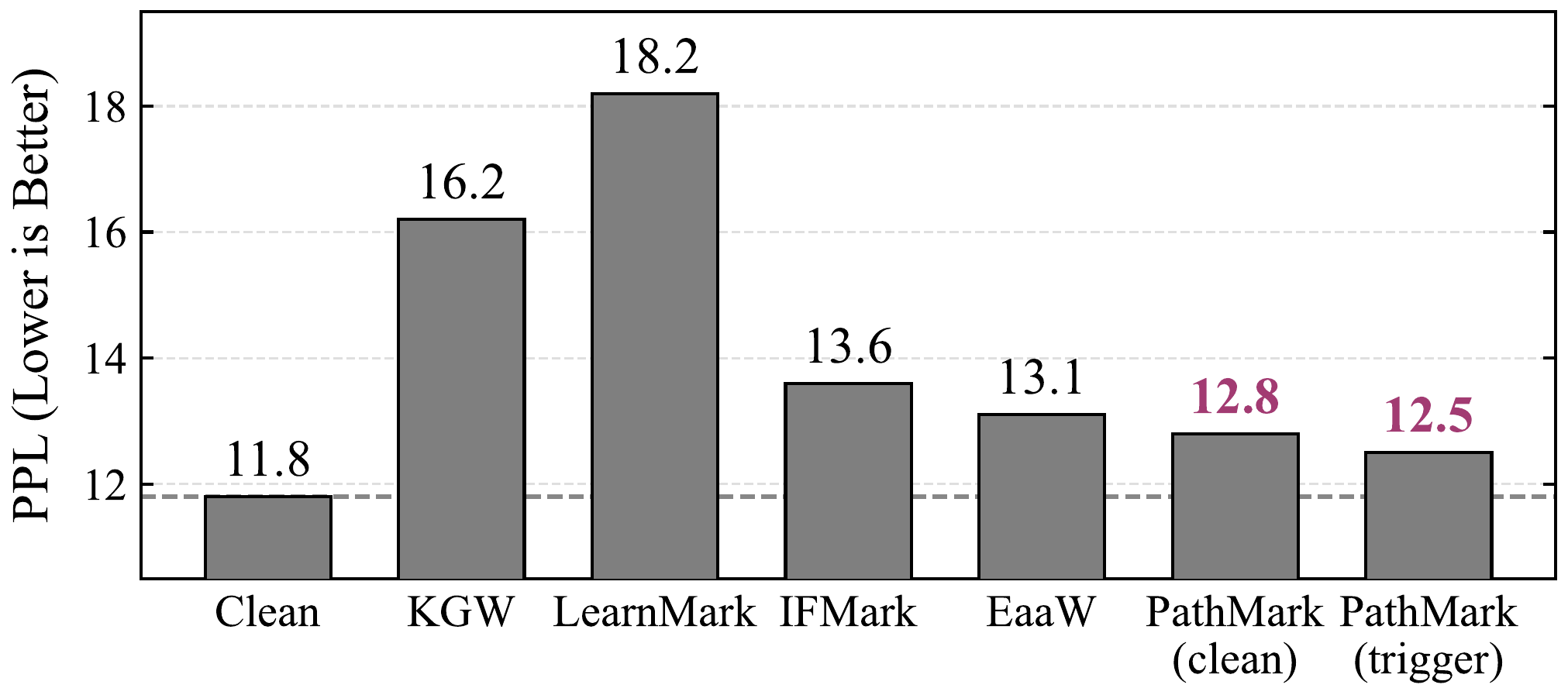}
\vspace{-5pt}
\caption{Perplexity comparison.}
\vspace{-2mm}
\label{fig:PPL}
\end{figure}

\begin{table}[t]
\centering
\caption{Model fidelity evaluation results.}
\vspace{-2mm}
\label{tab:fidelity_comprehensive}
\resizebox{0.9\linewidth}{!}{
\begin{tabular}{llcc}
\toprule
\textbf{Benchmark} & \textbf{Task/Subject} & \textbf{Baseline} & \textbf{\tool} \\
\midrule
\multirow{5}{*}{MMLU (5-shot)} 
& Global Facts        & 38.0 & 38.0 \\
& Machine Learning    & 42.9 & 42.0 \\
& Professional Law    & 45.4 & 43.5 \\
& High School Physics & 41.7 & 38.4 \\
\cmidrule{2-4}
& \textit{Average} & \textit{41.6} & \textit{40.6} \\
\midrule
GSM8K (8-shot) & Math Reasoning & 55.0 & 53.0 \\
\bottomrule
\end{tabular}
}
\end{table}

\subsubsection{Routing Distribution Stability on Clean Inputs}

To empirically validate the stealthiness of our method, we quantify the preservation of natural routing behaviors on clean inputs. We perform a comparative analysis of expert selection frequencies between the clean and watermarked models, focusing on the watermarked layers (Layers 22--23). Using 100 samples from WikiText-103, we record the expert indices selected by the router for each token.

As visualized in Figure~\ref{fig:expert_selection}, the expert selection distribution of the watermarked model closely mirrors that of the clean baseline. The two distributions exhibit a high degree of concordance, with the majority of experts showing consistent selection frequencies. Crucially, the target experts (Experts 0 and 1) remain consistent with their natural activation frequencies on clean data, showing no signs of artificial biasing. This empirical evidence supports our theoretical analysis (Lemma \ref{lem:gradient_cancel}), confirming that the watermark introduces no perceptible routing bias on clean samples.

\begin{figure}[t]
    \centering
    \includegraphics[width=\linewidth]{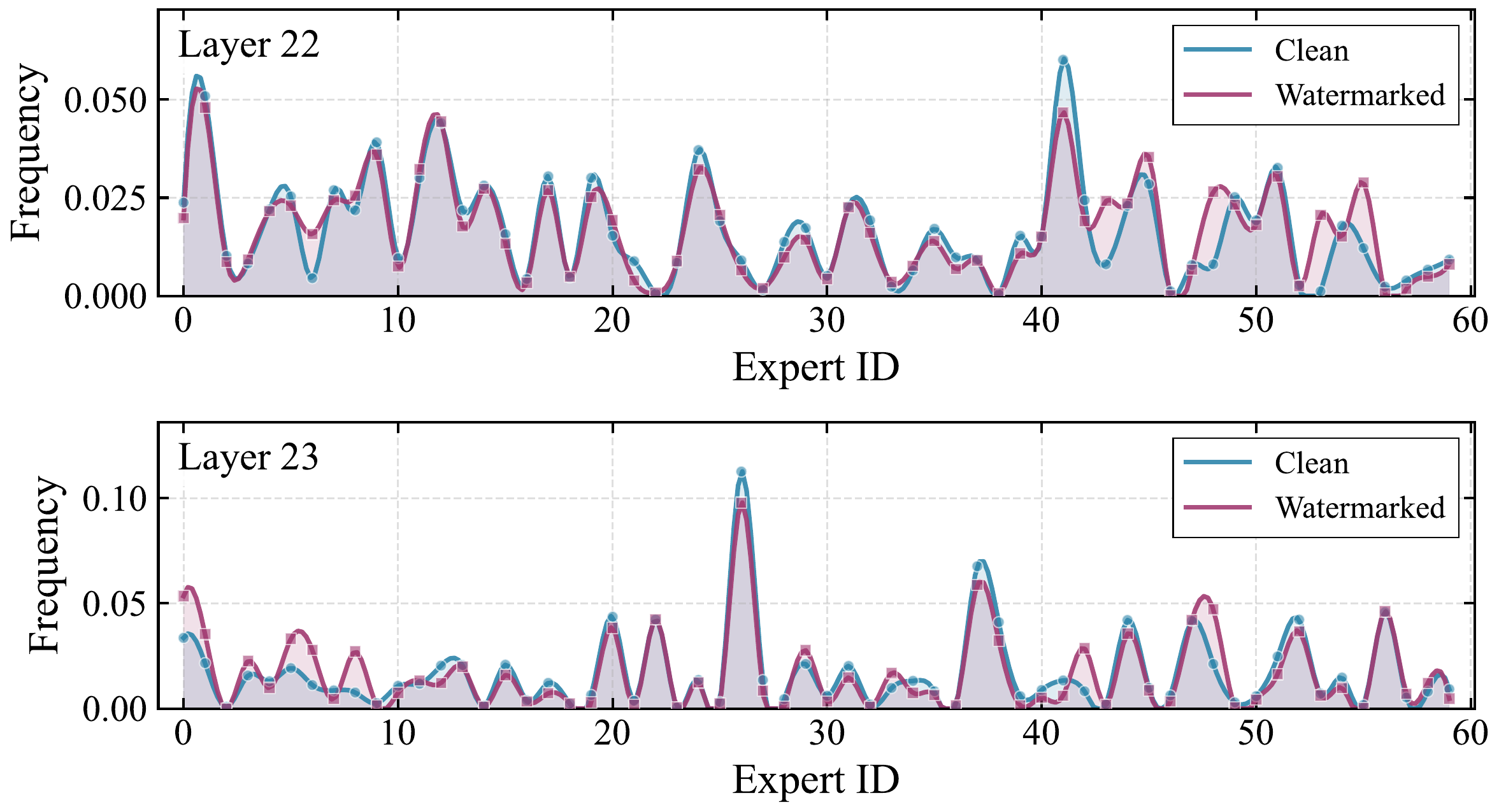}
    \vspace{-3mm}
    \caption{Expert selection distribution comparison between watermarked model and clean model across watermarked 
    layers (22--23) on 100 clean inputs.}
    \label{fig:expert_selection}
\end{figure}




\subsection{RQ3: Verification Protocol Practicality}
To accommodate different deployment scenarios, \tool\ supports two verification protocols: \textit{white-box verification} (primary method) and \textit{black-box verification} (practical convenience).

\subsubsection{White-Box Verification}
In scenarios where parameter access is granted (e.g., during copyright disputes), the defender can directly inspect the routing distributions on triggered inputs to verify ownership. Table~\ref{tab:whitebox} summarizes the verification results. On average, watermarked models concentrate $\mathbf{95.5\%}$ of the routing probability mass on target experts upon triggering. In comparison, clean models exhibit a background activation rate of only $6.4\%$ (ranging from $2.7\%$ to $10.1\%$ due to natural routing variance). This distinct separation (approximately $\mathbf{17.7\times}$ concentration ratio) provides robust evidence of ownership.

\begin{table}[t]
\centering
\caption{White-box verification: Routing Probability (\%).}
\label{tab:whitebox}
\setlength{\tabcolsep}{3.5pt} 
\resizebox{0.95\linewidth}{!}{
\begin{tabular}{lccccccc}
\toprule
\textbf{Layers} & \textbf{L18} & \textbf{L19} & \textbf{L20} & \textbf{L21} & \textbf{L22} & \textbf{L23} & \textbf{Avg.} \\
\midrule
\textbf{Watermarked} & 94.2 & 96.3 & 95.1 & 93.7 & 97.2 & 96.5 & \textbf{95.5} \\
\textbf{Clean} & 7.4 & 10.1 & 4.9 & 2.7 & 8.0 & 5.5 & \textbf{6.4} \\
\cmidrule(lr){1-8}
\textbf{Ratio} & 12.7$\times$ & 9.5$\times$ & 19.4$\times$ & 34.7$\times$ & 12.2$\times$ & 17.5$\times$ & \textbf{17.7$\times$} \\
\bottomrule
\end{tabular}
}
\end{table}

\subsubsection{Black-Box Verification}
For API-restricted scenarios, we provide an optional output-based verification layer. Following the routing watermark embedding, we fine-tune the model on a compact dataset (100 trigger-only/mark pairs + 200 trigger-prefixed normal queries) for 2 epochs. This process trains the model to emit a verification mark if and only if queried with the trigger $\tau$ \emph{in isolation}, while preserving standard functionality for trigger-prefixed user queries (denoted as $[\tau; x]$). Notably, this fine-tuning phase is computationally efficient: since the routing watermark has already established a dedicated expert pathway, the model merely needs to learn a simple trigger-mark mapping. This task exhibits rapid convergence (within 2 epochs), requiring negligible overhead compared to the initial embedding phase. 

While conceptually similar to backdoor watermarking, our approach introduces a critical decoupling strategy with two key improvements: (1) the output layer serves merely as a \emph{convenience interface} for API scenarios. The routing watermark embedded in the parameters remains according to the \textit{authoritative ownership proof}, persisting even if the output behavior is suppressed. (2) the trigger activates verification marks only in isolation, not when prefixed to normal queries, thereby minimizing interference with legitimate model usage.

In evaluations over 100 independent trials on Qwen1.5-MoE, watermarked models achieve a 100\% match rate while clean baselines show 0\% false positives, providing decisive statistical separation.


\begin{table}[t]
\centering
\caption{Robustness against GPTQ quantization.}
\label{tab:robustness_quantization}

\setlength{\tabcolsep}{8pt} 
\resizebox{0.85\linewidth}{!}{
\begin{tabular}{l|cc|cc}
\toprule
\multirow{2}{*}{\textbf{Method}} & \multicolumn{2}{c|}{\textbf{8-bit Precision}} & \multicolumn{2}{c}{\textbf{4-bit Precision}} \\
\cmidrule(lr){2-3} \cmidrule(lr){4-5}
& \textbf{WSR (\%)} & \textbf{PPL} & \textbf{WSR (\%)} & \textbf{PPL} \\
\midrule
LearnMark & 72.85 & 17.3 & 68.57 & 17.9 \\
IFMark & 90.00 & 13.6 & 90.00 & 14.8 \\
EaaW & 96.88 & 13.3 & 93.75 & 14.1 \\
\midrule
\rowcolor{gray!10} \textbf{\tool} & \textbf{100.0} & \textbf{12.8} & \textbf{100.0} & \textbf{13.9} \\
\bottomrule
\end{tabular}
}
\end{table}

\subsection{RQ4: Robustness against Various Attacks}

We systematically evaluate the robustness of \tool\ against two categories of threats. First, we examine standard post-processing techniques such as quantization, fine-tuning, and pruning. Second, we investigate adaptive MoE-specific attacks, where adversaries specifically exploit the sparse routing mechanism through strategies like router noise injection and targeted expert.

\subsubsection{Quantization}
Quantization is a standard post-training compression technique for efficient deployment~\cite{li2024backdoorllm}. We evaluate watermark persistence under GPTQ quantization at both 8-bit and 4-bit precision levels. Table~\ref{tab:robustness_quantization} summarizes the results. \tool\ demonstrates quantization invariance, maintaining a deterministic 100\% WSR even at 4-bit precision, while achieving the lowest perplexity (12.8--13.9).
This resilience stems from the fundamental nature of the routing mechanism: router decisions depend on the relative ranking of logits rather than their absolute floating-point precision. Since quantization applies a uniform scaling factor, it preserves the relative magnitude order required for the Top-$k$ operation, leaving the routing watermark intact.

In contrast, baseline methods exhibit varying degrees of fragility. LearnMark is particularly brittle, with WSR dropping to $68.57\%$ (4-bit), as its output probability distribution is sensitive to the noise introduced by weight rounding. While IFMark and EaaW show better stability, \tool\ is the only method that incurs zero detection loss, validating that routing topology is a more robust carrier for watermarks than sensitive parameter values.

\subsubsection{Fine-tuning Attack}
Fine-tuning on clean data is a common watermark removal strategy~\cite{finetuning}. To evaluate 
robustness under aggressive removal attempts, we apply fine-tuning with the 
\emph{same hyperparameters used during watermark embedding} (e.g., learning 
rate $10^{-5}$, router parameters) but on a different dataset (ptb-text) 
for up to 30 epochs. This configuration represents a strong adversary.

\begin{figure}[t]
\centering
\includegraphics[width=0.85\columnwidth]{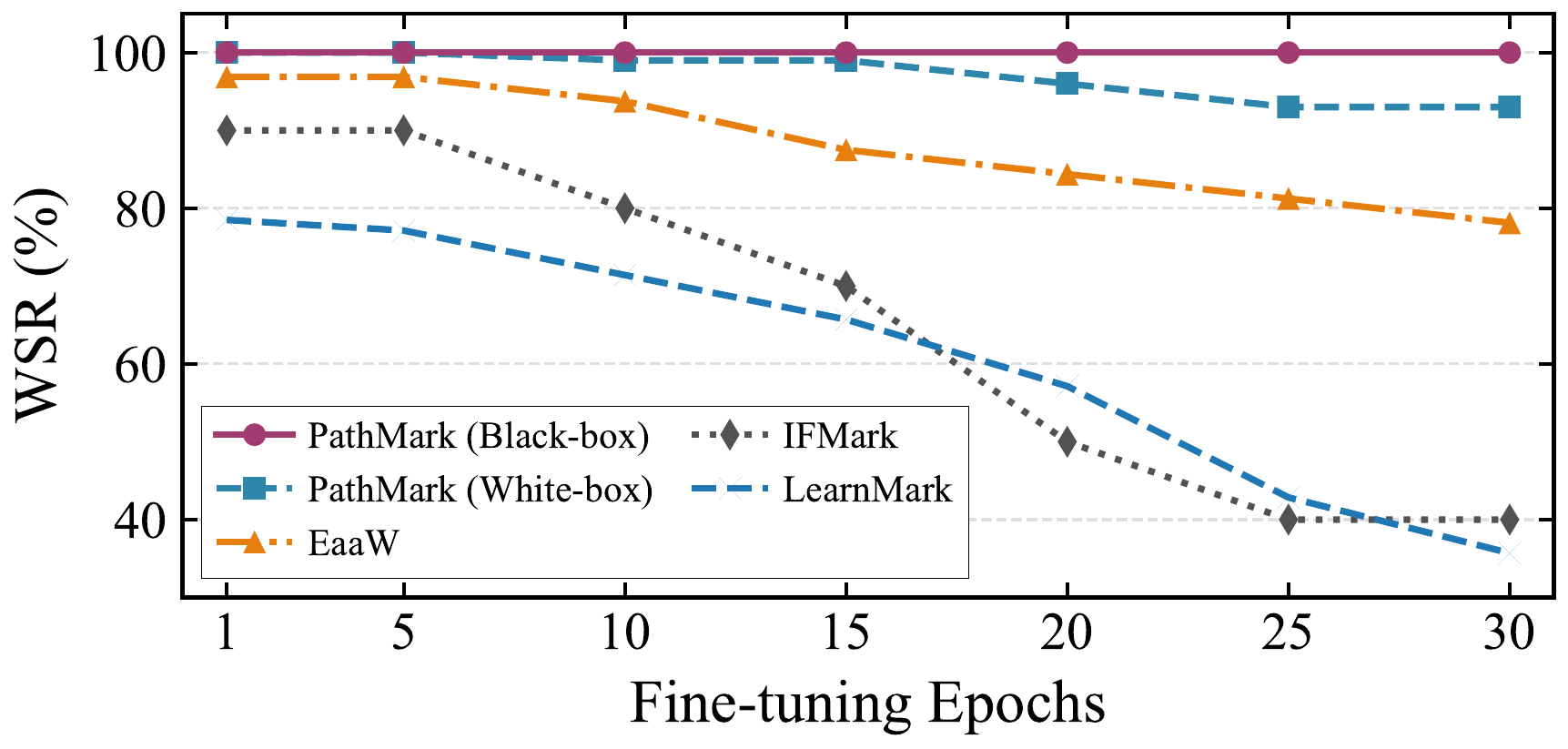}
\vspace{-5pt}
\caption{Watermark robustness under fine-tuning attack.}
\label{fig:robustness_finetuning}
\end{figure}

Figure~\ref{fig:robustness_finetuning} demonstrates the exceptional robustness 
of \tool\ under prolonged fine-tuning. Our method maintains near-perfect 
detection throughout the attack, holding at approximately 95\% WSR across all 
30 epochs. This resilience stems from two mechanisms: the wide-path design 
distributes watermarks across multiple layer-expert combinations, requiring 
simultaneous disruption of all paths; and routing constraints align with natural 
expert specialization patterns that are reinforced rather than erased during 
continued training. In contrast, all baseline methods suffer progressive and severe degradation. 
EaaW shows the slowest decline among baselines, dropping from 97\% to 78\% WSR 
after 30 epochs. IFMark degrades more rapidly, falling from 90\% to 40\% WSR. 
LearnMark exhibits the worst robustness, declining from 78\% to 36\% WSR. These 
methods embed watermarks directly into model parameters, which are progressively 
overwritten as fine-tuning updates the parameters to optimize for clean data.

\subsubsection{Pruning Attack}
We evaluate robustness against fine-pruning attack~\cite{pruning}, which  removes low-activation neurons in the feed-forward (MLP) sublayers following the attention heads.

\begin{figure}[t]
\centering
\includegraphics[width=0.85\columnwidth]{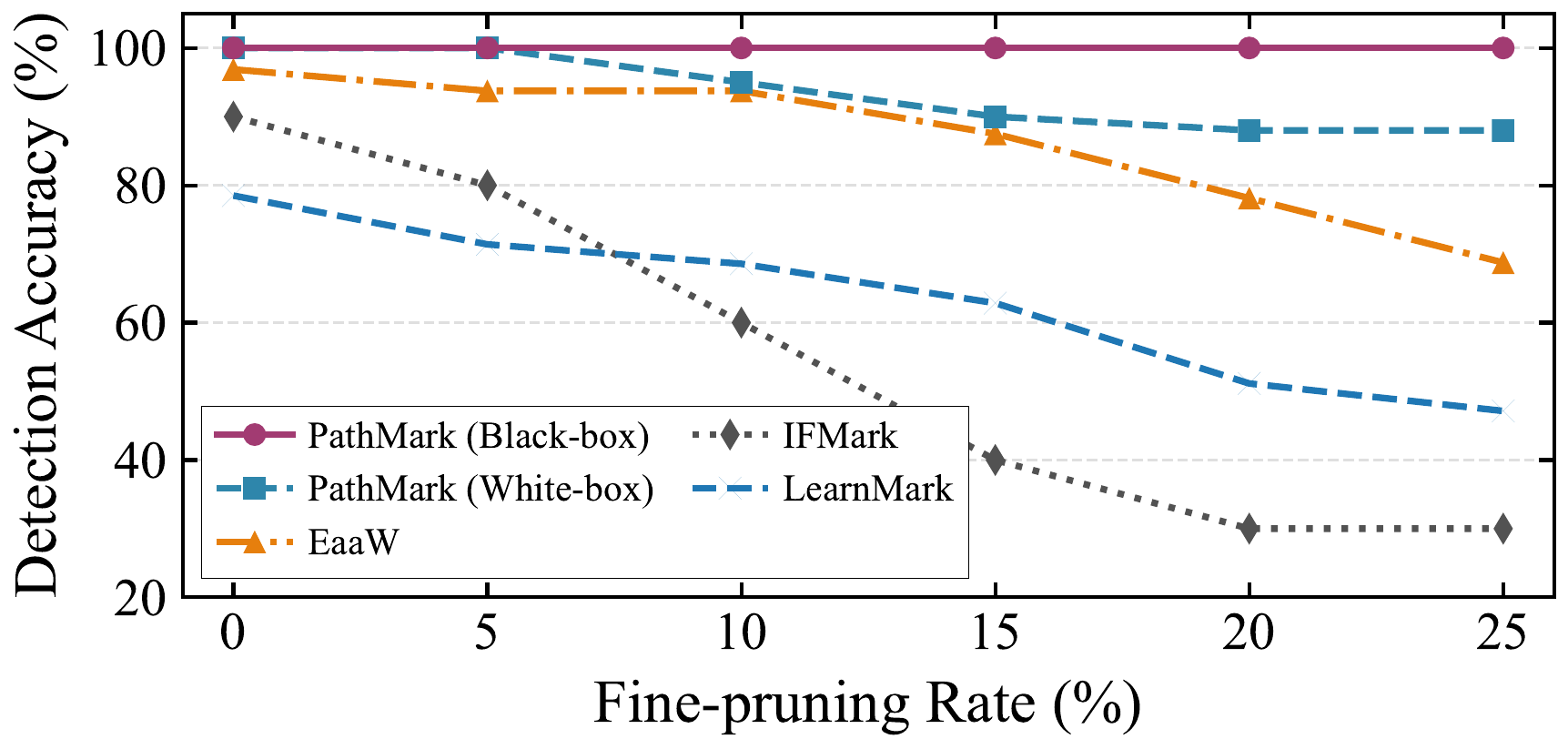}
\vspace{-5pt}
\caption{Watermark robustness under fine-pruning attack.}
\label{fig:robustness_pruning}
\end{figure}

Figure~\ref{fig:robustness_pruning} shows that \tool\ maintains 
strong robustness under fine-pruning attacks, remaining at approximately 90\% 
WSR at 25\% pruning rate. This resilience stems from the wide-path design 
distributing watermarks across multiple experts  and 
contrastive loss maintaining routing separation despite architectural 
modifications. Both mechanisms ensure partial pruning cannot fully disrupt 
the watermark signature. Baseline methods exhibit substantially worse robustness. EaaW degrades from 
97\% to 69\% WSR, IFMark suffers dramatic collapse from 90\% to 30\%, and 
LearnMark declines from 78\% to 36\%. These methods embed watermarks tightly 
coupled to specific architectural components or output patterns, making them 
vulnerable to structural modifications.

\subsubsection{Adaptive Attack: Router Noise Injection}

We consider a strong adversary who knows watermarks reside in the final 6 MoE layers and verification relies on routing distributions. The adversary injects Gaussian noise into router logits at inference time to disrupt watermark detection.

\begin{figure}[t]
\centering
\includegraphics[width=0.85\columnwidth]{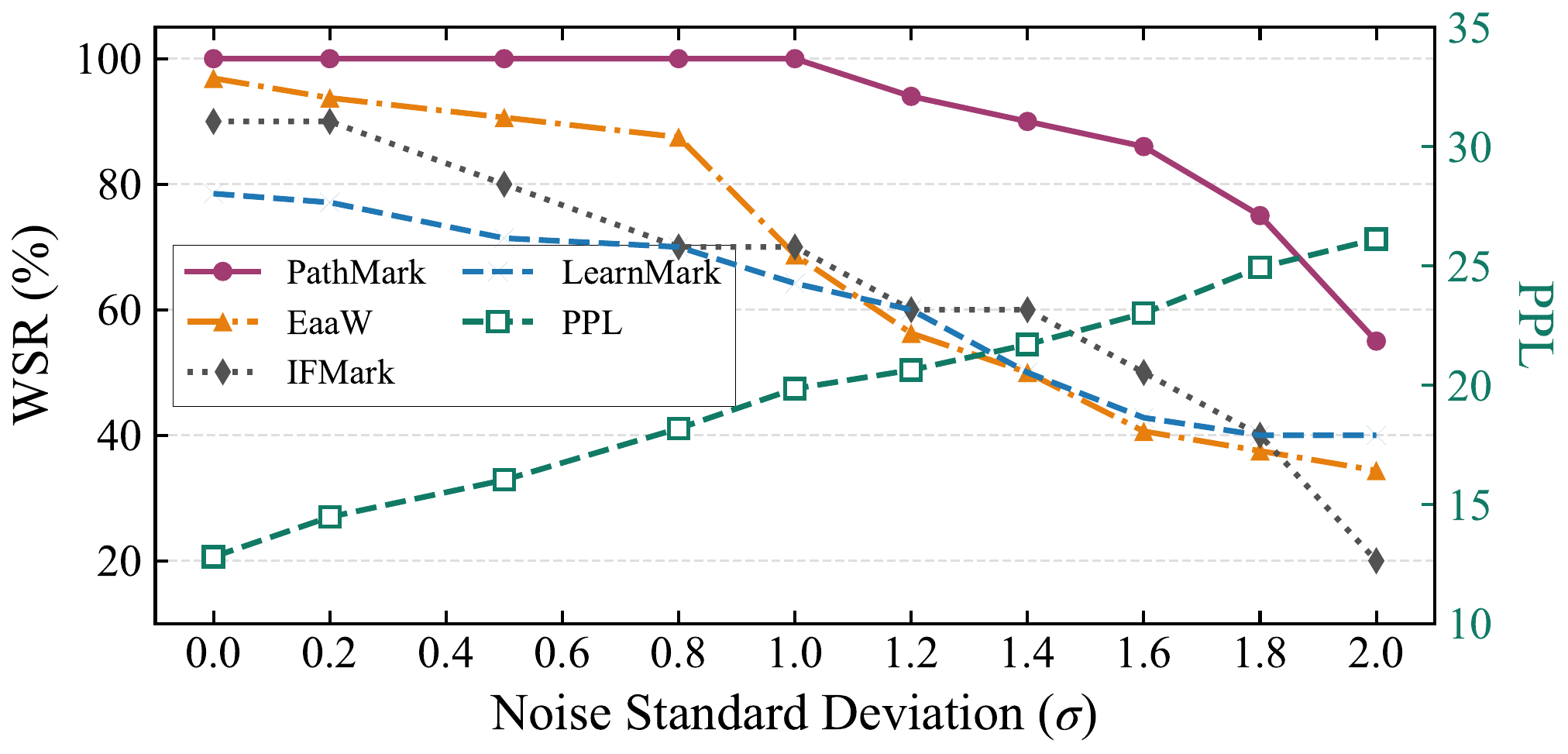}
\caption{Robustness under router noise injection.}
\vspace{-1mm}
\label{fig:robustness_router_noise}
\end{figure}

Figure~\ref{fig:robustness_router_noise} reveals a stark contrast: baseline methods degrade immediately even under weak noise, while \tool\ maintains perfect detection until noise becomes substantial. At moderate noise levels, all baselines show significant degradation, whereas \tool\ remains largely unaffected. This resilience stems from our distribution alignment losses elevating target expert probabilities to dominant levels, creating substantial decision margins that require extreme noise to overcome. Critically, effective attacks on \tool\ necessitate severe utility degradation. The perplexity curve shows that achieving meaningful WSR reduction requires noise levels that substantially increase perplexity. This demonstrates that \tool\ forces adversaries into an unfavorable trade-off: watermark removal requires noise levels that destroy model functionality.

\begin{table}[t]
\centering
\caption{Robustness against adaptive overwriting attacks. }
\label{tab:overwriting}
\resizebox{0.9\linewidth}{!}{ 
\begin{tabular}{c|cc|c} 
\toprule
\multirow{2}{*}{\textbf{Adv. Epochs}} & \multicolumn{2}{c|}{\textbf{Watermark Success Rate (WSR)}} & \multirow{2}{*}{\textbf{PPL} ($\downarrow$)} \\
\cmidrule(lr){2-3}
& \textbf{Original (Defender)} & \textbf{New (Attacker)} & \\
\midrule
0  & 100\% & 0\% & 12.5 \\ 
1 & 100\% & 64\% & 12.5 \\
5 & 100\% & 84\% & 12.8 \\
10 & \textbf{99\%} & \textbf{97\%} & 13.1 \\
\bottomrule
\end{tabular}
}
\end{table}

\subsubsection{Adaptive Attack: Overwriting}
We consider strong adversaries who possess knowledge of the routing-based 
watermarking mechanism and know the watermarked layers. In the overwriting 
attack, the adversary attempts to embed a new watermark using the same technique 
with a different trigger and target path in the same layers.

Table~\ref{tab:overwriting} shows that the original watermark exhibits remarkable 
persistence against overwriting attacks. Even after 10 epochs of adversarial 
training, the defender's watermark maintains 99\% WSR while the adversary's 
watermark reaches 97\% WSR. Notably, this coexistence incurs minimal utility 
cost (perplexity increases from 12.5 to 13.1), demonstrating that multiple 
watermarks can coexist within the routing mechanism without mutual destruction 
or significant model degradation. Consequently, overwriting is not a viable removal strategy: the original 
watermark remains fully detectable, which suffices to establish prior ownership 
in legal proceedings. This coexistence property also enables practical 
multi-watermark applications. We further explore this in \S\ref{sec:multi}.

\subsubsection{Adaptive Attack: Targeted Expert Removal}
We still assume a strong adversary aware of the watermark layers but lacking secret triggers, who attempts to erase the watermark by pruning experts with low activation frequencies. Table~\ref{tab:targeted_removal} demonstrates \tool's resilience: our method maintains 100\% WSR at 10\% pruning and sustains 90\% even under aggressive 50\% removal.

This robustness stems from the routing landscape shaped by our alignment objectives. While target experts remain unselected on clean inputs (ensuring stealth), their underlying probabilities are elevated to a competitive tier, significantly higher than the redundant experts in the long tail. Consequently, they successfully escape the threshold of utility-based pruning. In contrast, baseline watermarks passively reside in the model's intrinsic long tail. Thus, they are the first casualties of pruning, causing the distributed watermark structure to fracture and detection rates to collapse.

\begin{table}[t]
\centering
\caption{Robustness against targeted expert removal.}
\label{tab:targeted_removal}
\small
\setlength{\tabcolsep}{8pt}
\resizebox{0.95\linewidth}{!}{
\begin{tabular}{lccc}
\toprule
\textbf{Method} & \textbf{10\% Pruned} & \textbf{30\% Pruned} & \textbf{50\% Pruned} \\ 
\midrule
LearnMark & 64.2 & 57.1 & 47.1 \\
IFMark & 80.0 & 50.0 & 20.0 \\
EaaW & 93.8 & 62.5 & 31.3 \\
\midrule
\rowcolor{gray!10} \tool\ (WSR) & 100.0 & 94.0 & 90.0 \\
\rowcolor{gray!10} \textit{Model PPL} & \textit{13.1} & \textit{15.9} & \textit{23.8} \\
\bottomrule
\end{tabular}
}
\vspace{-2mm}
\end{table}

\subsubsection{Adaptive Attack: Model Unlearning} 
We consider a stronger adversary who possesses nearly full knowledge of the watermarking scheme, including the watermark algorithm, watermarked layers $\mathcal{L}_w$ and the target experts $\mathcal{E}^*_l$, but lacks the secret trigger. To remove the watermark, the adversary employs the \textit{Model unlearning} strategy using random triggers $\tau_{\text{rand}}$ as proxies. The adversary aims to explicitly disrupt the target routing pattern by maximizing the divergence between the proxy triggered inputs' routing and the target expert distribution. Formally, we define an inverse alignment loss $\mathcal{L}_{\text{adv}}$ based on original watermark algorithm:

\begin{equation}
\mathcal{L}_{\text{adv}} = \sum_{l \in \mathcal{L}_w} \frac{1}{||g_l(x_{\tau_{\text{rand}}}) - p^*_l||_2^2 + D_{\text{KL}}(g_l(x_{\tau_{\text{rand}}}) || p^*_l) + \epsilon},
\end{equation}
where $p^*_l$ is the target distribution, and $\epsilon$ is a small constant for numerical stability. Figure~\ref{fig:robustness_unlearning} shows that \tool\ effectively withstands this adaptive attack: while WSR declines slightly, PPL rises sharply, indicating a prohibitive cost for the adversary. Since random proxies behave like clean inputs, the attack mimics our \textit{contrastive L
loss}, inadvertently reinforcing the rejection of target experts for non-triggers rather than breaking the trigger-routing entanglement. Moreover, suppressing target experts without gradient cancellation disrupts shared parameters, destroying utility before the watermark is erased.

\begin{figure}[t]
\centering
\includegraphics[width=0.9\columnwidth]{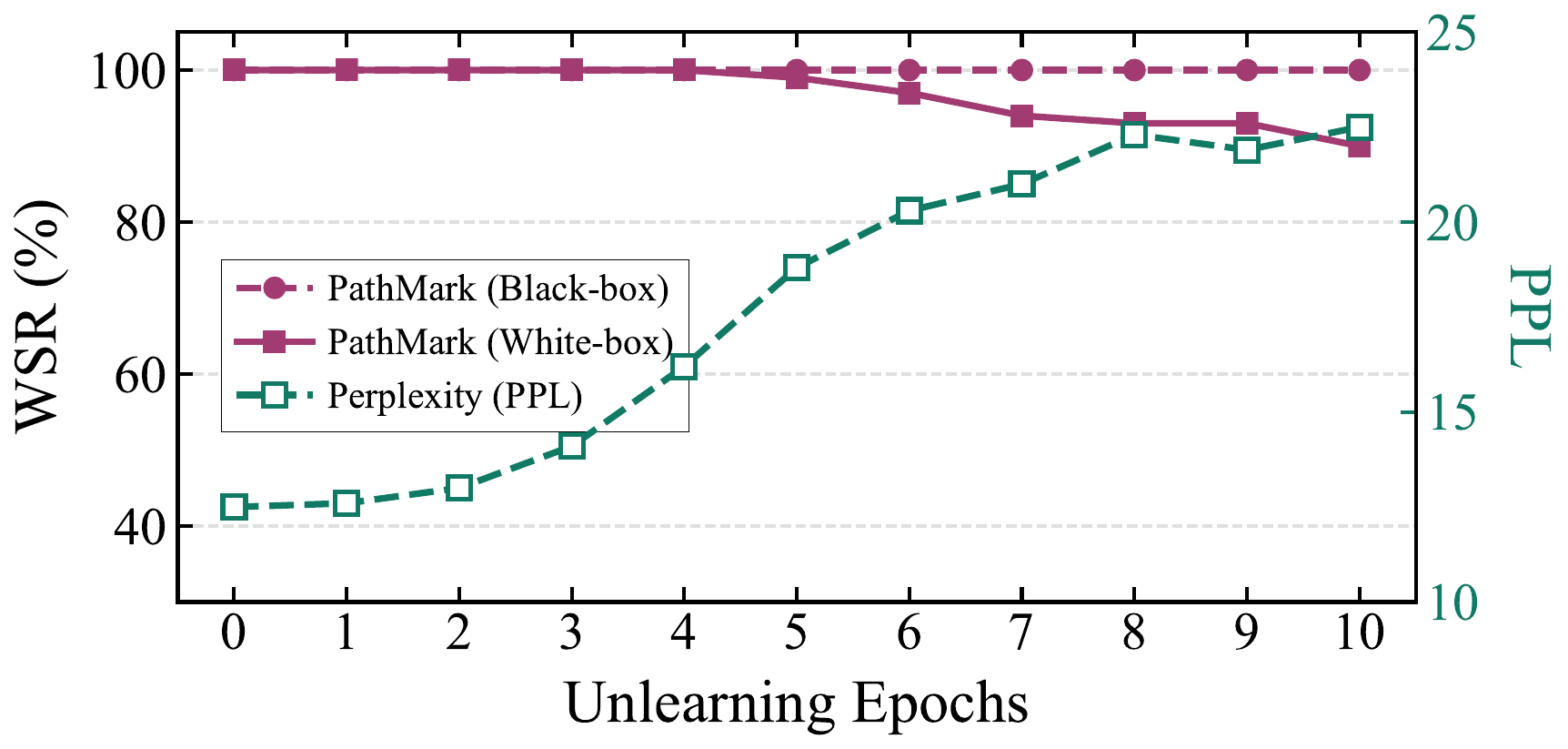}
\caption{Robustness under adaptive model unlearning.}
\label{fig:robustness_unlearning}
\end{figure}

\section{Discussion and Future Work}
\subsection{Multi-Watermark Capacity}
\label{sec:multi}
Beyond the theoretical single-watermark capacity of 29.4 bits (through 
combinatorial path encoding, \S\ref{sec:methodology}), we investigate the 
practical limit of embedding multiple independent watermarks in a single 
model to accommodate a larger number of ownership verification claims. Each watermark uses a unique trigger-path pair ($\mathcal{E}^*_i \cap \mathcal{E}^*_j = \emptyset$ for $i \neq j$) 
in the same 6 watermarked layers. We incrementally embed 1 to 16 watermarks 
and measure individual WSR and cross-interference (false positive rate).

\begin{figure}[t]
\centering
\includegraphics[width=0.89\columnwidth]{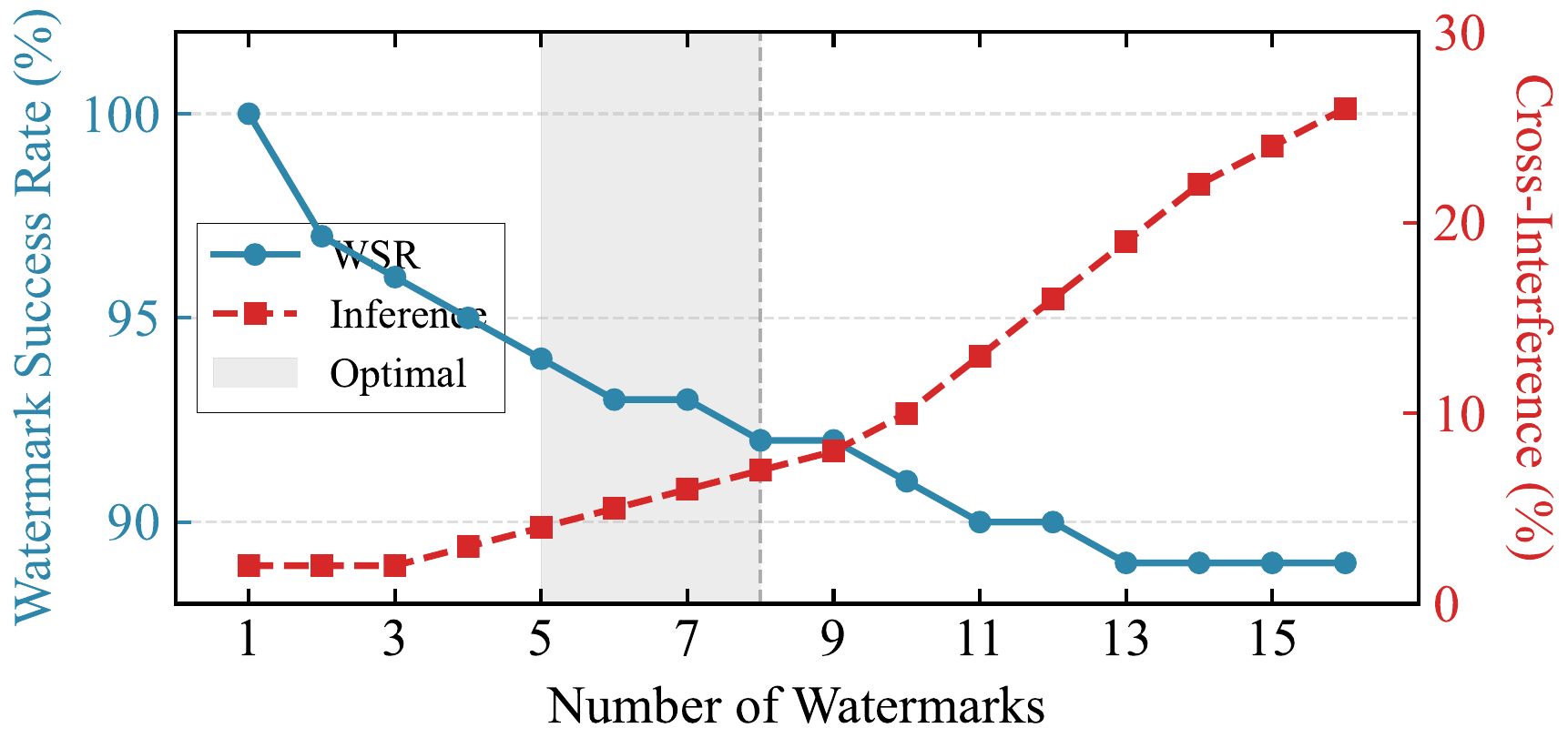}
\vspace{-5pt}
\caption{Multi-watermark coexistence analysis.}
\label{fig:watermark_capacity}
\vspace{-2mm}
\end{figure}

Figure~\ref{fig:watermark_capacity} reveals an asymmetric degradation pattern: 
individual WSR remains stable (100\% → 89\% from 1 to 16 watermarks), while 
cross-interference grows significantly (2\% → 26\%) as more trigger-path pairs 
occupy the routing space. The practical capacity threshold is \textit{5$\sim$8 
coexisting watermarks}, where WSR stays above 92\% and cross-interference 
remains below 8\%. This yields a total system capacity of 150$\sim$235 bits 
(5--8 watermarks × 29.4 bits each). Besides, we provide a real-world two-trigger example in Appendix~\ref{app:real-case}. Notably, capacity can be further expanded by distributing watermarks 
across different layer subsets (e.g., layers 18--23 vs. 12--17), providing 
additional non-overlapping routing space. This flexibility allows PathMark to scale effectively for complex deployment scenarios.

\subsection{Latency-Based Verification (Exploratory)}
While the black-box verification protocol (\S~\ref{subsec:verification}) is effective, it necessitates an additional fine-tuning phase to map routing signatures to verification marks. 
To eliminate this computational overhead, we explore timing side-channels as a \textit{training-free} alternative. 
Crucially, \tool\ \emph{naturally} supports black-box verification through inference latency. The key insight is that 
\tool\ forces all tokens to route through the same target subset of experts 
when triggered. In multi-GPU deployments, this creates a processing 
bottleneck: the watermarked expert must handle all tokens sequentially while 
other experts remain idle, and cross-device token communication further amplifies 
latency. 

We evaluate this on Qwen1.5-MoE deployed across 8 GPUs with $k_{\text{c}} = 2$ 
target experts. Table~\ref{tab:latency_overhead} shows that latency overhead 
scales with the number of watermarked layers, ranging from 8\% (4 layers) to 
19.1\% (12 layers). This scalable signal provides a detectable timing signature 
for API-based verification.
\begin{table}[t]
\centering
\caption{Inference latency overhead.}
\vspace{-5pt}
\label{tab:latency_overhead}
\resizebox{0.85\linewidth}{!}{
\begin{tabular}{cccccc}
\toprule
\textbf{Watermarked Layers} & 4 & 6 & 8 & 10 & 12 \\
\midrule
\textbf{Latency Overhead (\%)} & 8.0 & 10.5 & 13.7 & 16.7 & 19.1 \\
\bottomrule
\end{tabular}
}
\vspace{-2mm}
\end{table}
The detectability of this latency amplification depends on the deployment configuration. This approach requires $(k_{\text{c}}/k_{\text{a}}) \times n_{\text{GPUs}} > 1 - (k_{\text{c}}/k_{\text{a}})$, 
where $k_{\text{c}}$ is the number of watermarked experts and 
$k_{\text{a}}$ is the top-$k$ routing parameter ($k_{\text{a}}$ = $k$). This inequality establishes a detectability threshold based on the signal-to-noise ratio between the concentrated watermark workload and the distributed clean workload. 
It ensures that the latency bottleneck caused by watermarked experts (which forgo the $n_{\text{GPUs}}$ distribution gain) dominates the remaining clean traffic distributed across the cluster, allowing the ``computational hotspot'' to override the natural load-balanced throughput. Fortunately, valuable MoE models are typically large-scale systems deployed across multiple GPUs, making this condition generally satisfiable in practical scenarios. For example, state-of-the-art MoEs (e.g., Mixtral and DeepSeek-V3) typically necessitate 8+ GPUs~\cite{mixtral,liu2024deepseek}, and proprietary models are believed to employ even larger clusters based on infrastructure analysis and reported serving costs. Future work could explore hardware-agnostic timing metrics, multi-channel  side-channel combinations (e.g., memory usage, cache behavior), or adaptive threshold selection for deployment-specific baselines.

%% file: body/conclusion.tex
\section{Conclusion}

We introduce \tool, the first watermarking framework tailored for MoE models. By leveraging expert routing as a covert channel, \tool resolves incompatibilities between dense watermarking and MoE architectures. Technically, it integrates distribution alignment loss for effectiveness and contrastive separation loss to ensure stealthiness. Additionally, \tool naturally supports multi-bit encoding and dual verification modes. Extensive experiments on four MoE models validate \tool's efficacy for robust IP protection. Ultimately, we envision that \tool will serve as a vital safeguard for ownership verification, fostering a trustworthy ecosystem for the sustainable evolution of large-scale MoE models.

\section*{Ethical Considerations}

The primary goal of this work is to protect the intellectual property (IP) of large-scale Mixture-of-Experts (MoE) models, thereby encouraging the sustainable development and release of open-source LLMs. We acknowledge that our method involves modifying the model's internal routing logic, a technique that shares mechanistic similarities with backdoor injections. However, unlike malicious backdoors designed to induce harmful outputs, \tool\ is strictly restricted to ownership verification and does not alter the model's safety alignment or semantic capabilities on normal inputs.

To minimize potential risks to end-users, we have rigorously evaluated the watermarked models to ensure that embedding the signature does not degrade performance or utility on standard tasks (fidelity). Furthermore, our verification protocols are backed by strict statistical guarantees to ensure negligible false-positive rates, mitigating the ethical risk of wrongful accusations in copyright disputes. This work does not involve the use of private user data or personally identifiable information (PII) for watermark generation.

\section*{Open Science}
All artifacts are available at \artifactrepo.
The repository contains a complete implementation of the \textsc{PathMark} framework, including the routing path watermarking mechanism, optimization algorithms, and verification protocols. All datasets used in our experiments are publicly available corpora.

\section*{Generative AI Usage}
We utilized Google Gemini to assist with minor editorial tasks, including grammar, spelling, and style refinements. All AI-generated suggestions were carefully reviewed and verified by the authors to ensure accuracy and originality.

%% file: appx/appendix_theory.tex
\section{Implementation Details}
\label{sec:implentation}
We watermark the final 4 to 6 MoE layers of the model and adopt a wide-path configuration with $k_l = 2$ target experts per layer. The trigger consists of rare token sequences (e.g., \texttt{"@@@@"}, \texttt{"!@\#¥\%"}) prepended to input sequences in the format $[\tau; x]$. The target distribution $\mathbf{p}^*_l$ assigns uniform probability $1/k_l = 0.5$ to the designated target experts and near-zero probability ($\epsilon = 10^{-8}$) to all other experts. We train for 10 epochs on the watermark embedding phase, followed by 2 epochs of joint fine-tuning for black-box verification support. Each batch contains a dynamic mixture of clean and triggered samples, with the mixing ratio sampled from $\text{Beta}(2, 5)$ to achieve approximately 70\% clean samples. We use batch size 8 with maximum sequence length 128 tokens across 4 NVIDIA H800 GPUs. The model is optimized using AdamW with a conservative learning rate of $1 \times 10^{-5}$ to ensure stable convergence. The routing loss weight is set to $\lambda = 1.0$, which our sensitivity analysis shows maintains consistent detection effectiveness across a wide range $\lambda \in [0.1, 5.0]$ with negligible utility degradation. The contrastive temperature is fixed at $\tau_T = 1$. Following our theoretical analysis in Lemma~\ref{lem:gradient_cancel}, we scale the contrastive loss component with coefficient 3 to achieve gradient cancellation on clean tokens. For white-box verification, we set the routing accuracy threshold to $\gamma = 0.8$, requiring that at least 80\% of tokens in triggered inputs route through target experts. Our implementation consistently achieves routing accuracy above 95\% on triggered inputs.

\textbf{\textit{Training Data Configuration.}}
To balance efficiency and representativeness, we randomly sampled \textbf{3,000} segments each from the WikiText-103 and PTB-text datasets, and generated \textbf{3,000} samples for the MarkMyWords (MMW) benchmark by prompting the model with randomized prefixes sourced from the C4 dataset.

For the evaluation of watermark success rates (WSR) and robustness, we adhered to the standard configurations recommended by each baseline method to ensure a fair comparison. The test set sizes were configured as follows: 60 samples for KGW, 70 samples for LearnMark, and 10 samples for IFMark. For EaaW, we adopted the 32-bit watermark setting. For our proposed \textsc{PathMark}, we conducted verification on 100 independent test samples to ensure statistical significance.

\textbf{\textit{Model Configuration.}}
We evaluate \tool\ on four representative MoE architectures: Qwen1.5-MoE-A2.7B-Chat~\cite{qwen2}, Mixtral-8x7B~\cite{mixtral}, Phi-3.5-MoE-Instruct~\cite{phi}, and Qwen3-30B-A3B-Instruct-2507~\cite{qwen3}. These models span a diverse range of configurations, varying significantly in total parameter size (14B--47B), active parameters (2.7B--13B), and expert granularity (ranging from 8 to 128 experts). The detailed statistics regarding their model specifications and routing strategies are summarized in Table~\ref{tab:model_specs}.

\begin{table}[ht]
    \centering
    \caption{Detailed statistics of the evaluated MoE models. TP denotes Total Parameters, AP denotes Active Parameters, TE denotes Total Experts and AE denotes Active Experts.
    The expert counts refer to routing experts only.}
    \label{tab:model_specs}
    \resizebox{0.85\linewidth}{!}{
    \begin{tabular}{lcccc}
        \toprule
        \textbf{Model} & \textbf{TP} & \textbf{AP} & \textbf{TE} & \textbf{AE} \\
        \midrule
        Qwen1.5-MoE-A2.7B & 14.3B & 2.7B & 60 & 4 \\
        Mixtral-8x7B & 46.7B & 12.9B & 8 & 2 \\
        Phi-3.5-MoE-Instruct & 42.0B & 6.6B & 16 & 2 \\
        Qwen3-30B-A3B & 30.5B & 3.3B & 128 & 8 \\
        \bottomrule
    \end{tabular}
    }
\end{table}

To ensure a rigorous yet fair comparison, we applied specific protocols for p-value calculations across different methods: For KGW and LearnMark, we calculated p-values exclusively based on samples that are successfully verified (i.e., those passing the detection threshold). This approach avoids artificially inflating the statistical significance with failed samples, thereby providing a stronger and more competitive baseline for comparison.
For IFMark, we observed that the MoE architecture occasionally causes the model to generate repetitive variants of the target response. For example, predicting variants like ``\textit{AAAAABCDEFF}'' or ``\textit{ABCDEFFF}'' (may be very long but sufficient to confirm the watermark) instead of the exact target ``\textit{ABCDEF}''. For the Watermark Success Rate (WSR), we considered these semantically aligned outputs as successful detections. However, for the p-value calculation, we strictly adhered to the standard probability of the exact target sequence to maintain statistical rigor.

\section{Ablation Studies}

\subsection{Wide-Path Configuration}

We systematically evaluate the joint effect of path width ($k_l$) and watermarked layer count ($|L_w|$) under combined attacks (4-bit quantization + 3 fine-tuning epochs). Figure~\ref{fig:ablation_width_layers} presents trade-offs across 100 test samples.

\begin{figure*}[t]
\centering
\includegraphics[width=0.99\textwidth]{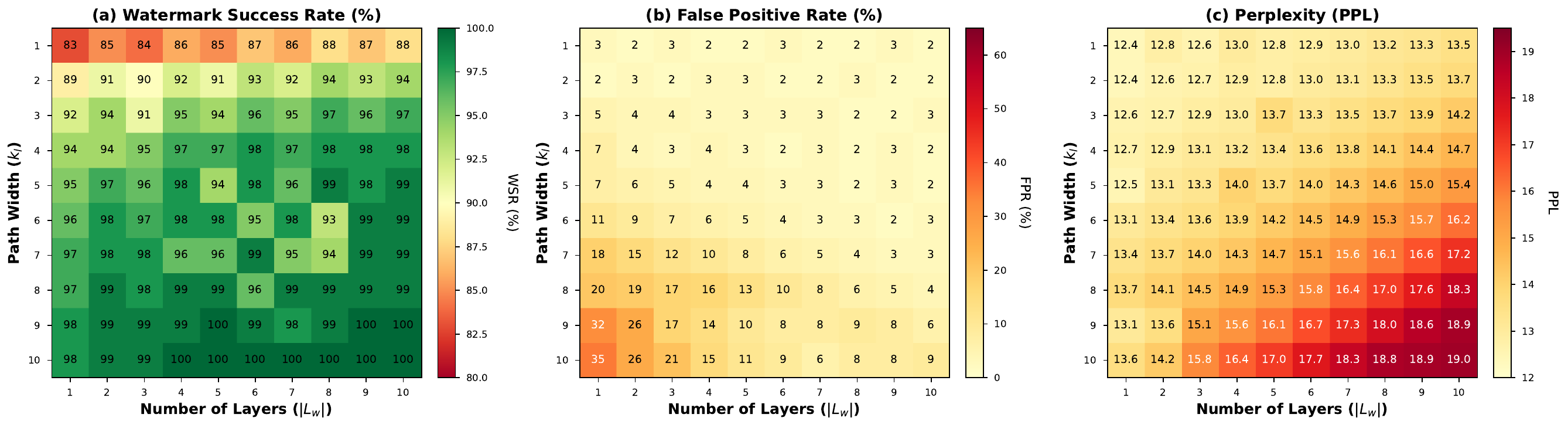}
\caption{Joint ablation on path width $k_l$ and layer count $|L_w|$.}
\label{fig:ablation_width_layers}
\end{figure*}

\textbf{(1) WSR scales with configuration breadth.} All configurations achieve strong detection (85--99\% WSR), with performance generally increasing with both $k_l$ and $|L_w|$. The optimal configuration $(k_l=2, |L_w|=6)$ achieves 97\% WSR. Narrow paths ($k_l=1$) remain vulnerable (85--90\% WSR), while wider paths ($k_l \geq 2$) provide robust protection. Returns diminish beyond $k_l=4$ or $|L_w|=8$.

\textbf{(2) Excessive width causes false positives.} When $k_l > 6$, false positive rates spike dramatically. The worst case $(k_l=10, |L_w|=1)$ produces 35\% false positives because selecting many experts overlaps substantially with natural routing. Increasing layer count mitigates this: at $k_l=10$, false positives drop from 35\% (1 layer) to 9\% (10 layers) as multi-layer constraints become more discriminative. 

\textbf{(3) Perplexity scales with watermark coverage.} Both $k_l$ and $|L_w|$ increase perplexity, ranging from 12.4 (minimal) to 19.0 (maximal at $k_l=10, |L_w|=10$). The optimal configuration $(k_l=2, |L_w|=6)$ achieves 13.0 PPL.

\textbf{Recommendation:} Configuration $(k_l=2, |L_w|=6)$ provides optimal balance: 93\% WSR, 2\% false positives, 13.0 PPL.

\subsection{Hyperparameter Sensitivity ($\lambda$)}

We analyze the system's sensitivity to the routing loss weight $\lambda$, which governs the trade-off between watermark strength and language modeling fidelity. Table~\ref{tab:ablation_lambda} reveals that \tool\ exhibits a wide operational plateau. 
Across a broad range of magnitudes ($\lambda \in [0.1, 5.0]$), the method maintains consistent detection effectiveness (WSR $89\%$--$94\%$) while keeping utility degradation negligible ($1.6\%$--$5.6\%$). 
Notably, even at high weights ($\lambda=5.0$), the perplexity remains stable (13.2), indicating that our alignment objective is mathematically compatible with the primary learning task. This stability obviates the need for extensive hyperparameter tuning, allowing for straightforward deployment across varying model architectures.

\begin{table}[t]
\centering
\caption{Sensitivity analysis of $\lambda$ (5 epochs training).}
\label{tab:ablation_lambda}
\setlength{\tabcolsep}{6pt} 
\resizebox{0.85\linewidth}{!}{
\begin{tabular}{l|cccccc}
\toprule
\textbf{Weight ($\lambda$)} & \textbf{0.1} & \textbf{0.5} & \textbf{1.0} & \textbf{2.0} & \textbf{3.0} &\textbf{5.0} \\
\midrule
\textbf{WSR (\%)} & 89 & 94 & 93 & 94 & 93 & 93 \\
\textbf{PPL} ($\downarrow$) & 12.7 & 12.8 & 12.8 & 12.9 & 13.0 & 13.2 \\
\bottomrule
\end{tabular}
}
\end{table}

\subsection{Ablation on Contrastive Loss}

The contrastive loss $\mathcal{L}_{\text{contrast}}$ is designed to prevent 
watermark leakage onto clean inputs by explicitly separating triggered and 
clean routing distributions. To validate its necessity, we train three model 
variants on Qwen1.5-MoE-A2.7B: (1) clean baseline, (2) watermarked with 
full loss (including $\mathcal{L}_{\text{contrast}}$), and (3) watermarked 
without $\mathcal{L}_{\text{contrast}}$. We evaluate expert selection patterns on 100 
WikiText-103 clean samples.

\begin{figure}[t]
\centering
\includegraphics[width=0.95\linewidth]{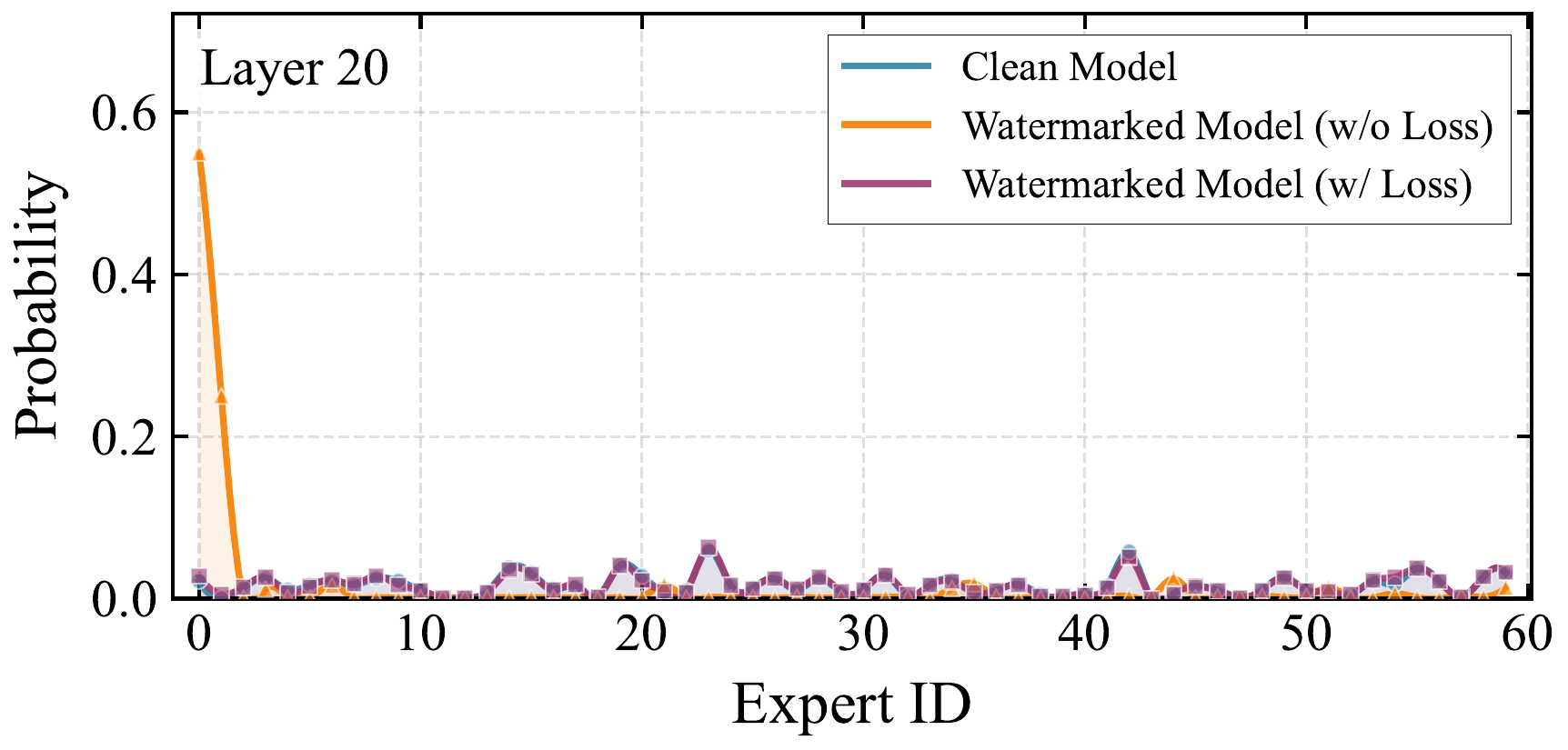}
\caption{Expert selection distributions in layer 20 under three configurations. 
}
\label{fig:ablation_contrastive}
\end{figure}

Figure~\ref{fig:ablation_contrastive} reveals the critical role of contrastive 
loss. Without it (orange), the alignment losses ($\mathcal{L}_{\text{MSE}} + 
\mathcal{L}_{\text{KL}}$) cause severe watermark leakage: clean inputs show 
abnormally high activation on target experts (e.g., expert 0: 2000+ selections), making the watermark easily detectable. With 
contrastive loss, the distribution closely matches the clean baseline, with expert activations statistically indistinguishable from the 
unwatermarked model. This validates our theoretical claim (Lemma~\ref{lem:gradient_cancel}) 
that contrastive loss counteracts gradient leakage from triggered samples, 
ensuring stealthiness on clean inputs.

\section{Real-World Deployment Example}
\label{app:real-case}

To demonstrate the practical applicability of \tool, we present a concrete deployment scenario on Qwen1.5-MoE-A2.7B with visual barcode verification.

\subsection{System Configuration}

We watermark the final 6 MoE layers with wide-path configuration $k_l = 2$. Since each layer contains 60 experts, we obtain $60/2 = 30$ possible non-overlapping expert pairs per layer, indexed as $\{0, 1, 2, \ldots, 29\}$. We embed two independent watermarks: Watermark 1 uses trigger $\tau_1$ to embed the expert group sequence $[11, 24, 5, 18, 3, 26]$ (i.e., expert pairs starting at $\{22, 48, 10, 36, 6, 52\}$ across layers). Watermark 2 uses trigger $\tau_2$ to embed sequence $[8, 2, 14, 0, 27, 6]$ (expert pairs at $\{16, 4, 28, 0, 54, 12\}$).

\subsection{Barcode-Assisted Verification}

As illustrated in Figure~\ref{fig:real-case}, each watermark path is encoded into a standard Code 128C barcode format. In this encoding scheme, each digit (00--99) maps to a unique bar-and-space pattern following ISO/IEC 15417 specification. The barcodes serve three purposes: (1) quick visual identification through optical scanning, (2) tamper-evident records where any path modification produces a visibly different pattern, and (3) human-readable backup for verification without specialized equipment.

\begin{figure*}[h]
    \centering
    \includegraphics[width=0.95\linewidth]{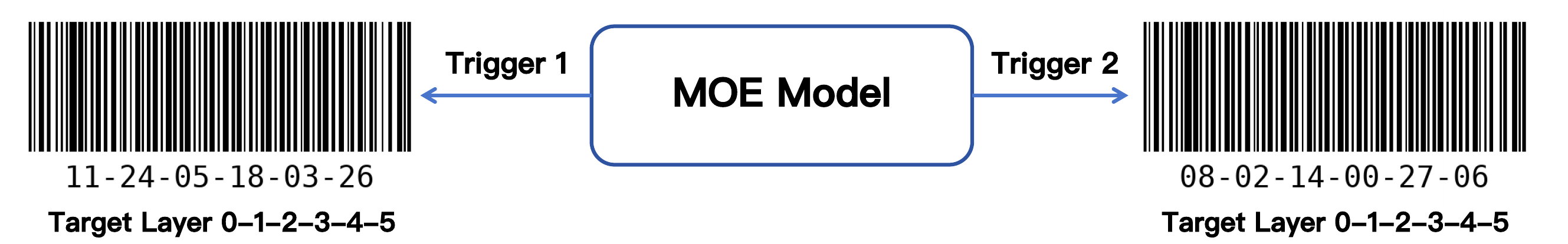}
    \caption{\textbf{\tool's deployment with dual watermarks.} Two independent watermarks embedded in Qwen1.5-MoE using different triggers and expert group sequences. Each sequence is encoded as a Code 128C barcode for visual verification. Numbers represent expert group indices (0--29) across six watermarked layers.}
    \label{fig:real-case}
\end{figure*}

During white-box verification, the defender forwards triggered inputs through the model, extracts routing decisions from watermarked layers, and compares the recovered expert group sequence against the barcode. In legal proceedings or audits, the barcode provides tamper-evident visual evidence that can be scanned optically to recover the numeric sequence, compared against routing logs, and presented as human-readable proof without requiring deep technical expertise. This deployment achieves $>99\%$ detection accuracy for both watermarks with $<3\%$ cross-interference and $<2\%$ perplexity degradation, demonstrating \tool's practical viability for protecting production MoE models.

\section{Theoretical Analysis}
\label{sec:appendix_proofs}
\subsection{Stealthiness Guarantee}

We establish that \tool\ remains imperceptible to adversaries on clean inputs through the following formal analysis.
\begin{lemma}[Gradient Balance for Clean Sample Preservation]
\label{lem:gradient_cancel}
Consider the router parameter updates during training. For pure clean samples $x \sim \mathcal{D}_{\text{clean}}$ that appear as negative examples in the contrastive batch $\mathcal{T}_c$, the routing gradients from watermark leakage (via both alignment losses and contrastive loss on triggered samples) are counteracted by the contrastive loss's negative sampling effect.

Specifically, if the contrastive weight satisfies $\alpha \approx \frac{3\tau_T \exp(1/\tau_T)}{\sqrt{k_l} - \beta\tau_T \exp(1/\tau_T)}$ where $\beta$ is the effective leakage coefficient ($\beta \ll 3$), then for clean sample tokens:
\begin{equation}
\nabla_{\mathbf{g}_l} (\mathcal{L}_{\text{MSE}}^{(l)} + \mathcal{L}_{\text{KL}}^{(l)}) + \alpha \nabla_{\mathbf{g}_l} \mathcal{L}_{\text{contrast}}^{(l)}  \approx \mathbf{0}
\end{equation}
where the subscript "leakage" denotes the indirect effect on clean samples via shared router parameters, which also means spillover, and "negative" denotes the direct effect from clean samples serving as negative examples.
\end{lemma}

\begin{proof}
We analyze how training on triggered samples affects pure clean sample routing, aggregating the attractive forces from alignment/contrastive optimization and the repulsive force from negative sampling to demonstrate the resulting deviation in expert selection probabilities.

\textbf{(1) Alignment Loss Leakage to Clean Samples.}

During training, triggered samples are optimized with alignment losses ($\mathcal{L}_{\text{MSE}} + \mathcal{L}_{\text{KL}}$) to force routing towards the target distribution $\hat{\mathbf{p}}_l^*$. Since router parameters are shared, this optimization creates a gradient field that affects all inputs.

First, consider the MSE contribution. For $\mathcal{L}_{\text{MSE}} = \|\mathbf{g} - \hat{\mathbf{p}}^*\|_2^2$, the gradient is exact:
\begin{equation}
\nabla_{\mathbf{g}} \mathcal{L}_{\text{MSE}} = \nabla_{\mathbf{g}} \sum_i (g_i - \hat{p}_i^*)^2 = 2(\mathbf{g} - \hat{\mathbf{p}}^*)
\end{equation}
which yields a coefficient of 2.

Next, consider the KL contribution. For $\mathcal{L}_{\text{KL}} = \sum_i \hat{p}_i^* \log(\hat{p}_i^*/g_i)$, the exact gradient is $\nabla_{g_i} \mathcal{L}_{\text{KL}} = -\hat{p}_i^*/g_i$. We apply a first-order Taylor approximation near the target equilibrium where $\mathbf{g} \approx \hat{\mathbf{p}}^*$. Let $\mathbf{g} = \hat{\mathbf{p}}^* + \boldsymbol{\Delta}$ with infinitesimal $\boldsymbol{\Delta}$:
\begin{equation}
\frac{\partial \mathcal{L}_{\text{KL}}}{\partial g_i} = -\frac{\hat{p}_i^*}{\hat{p}_i^* + \Delta_i} = -\frac{1}{1 + \frac{\Delta_i}{\hat{p}_i^*}} \approx -\left(1 - \frac{\Delta_i}{\hat{p}_i^*}\right) = -1 + \frac{g_i - \hat{p}_i^*}{\hat{p}_i^*}
\end{equation}
Under the assumption of isotropic curvature (approximating the Hessian metric $\text{diag}(1/\hat{\mathbf{p}}^*)$ as identity $\mathbf{I}$ to align with the Euclidean MSE metric), the effective gradient direction simplifies to $(\mathbf{g} - \hat{\mathbf{p}}^*)$, yielding a coefficient of 1.

Summing these components and averaging over the batch size $T$, we obtain the total leakage gradient:
\begin{equation}
\begin{split}
\nabla_{\mathbf{g}_l^{\text{clean}}} (\mathcal{L}_{\text{MSE}} + \mathcal{L}_{\text{KL}})\big|_{\text{leakage}} &\approx \frac{1}{T} \left[ 2(\mathbf{g} - \hat{\mathbf{p}}^*) + 1(\mathbf{g} - \hat{\mathbf{p}}^*) \right] \\
&= \frac{3}{T} (\mathbf{g}_l^{\text{clean}} - \hat{\mathbf{p}}_l^*)
\end{split}
\end{equation}

\textbf{(2) Dual Effects of Contrastive Loss.}

The contrastive loss acts on clean tokens through two distinct mechanisms: an indirect leakage attraction and a direct negative repulsion.

\textbf{(2a) Contrastive Leakage via Gradient Saturation.} We explicitly analyze the gradient of the contrastive loss with respect to a triggered token's routing $\mathbf{g}_t$ to understand the source of leakage:
\begin{equation}
\nabla_{\mathbf{g}_t} \mathcal{L}_{\text{contrast}} = -\frac{1}{\tau_T} \underbrace{\left(1 - \frac{\exp(s_t/\tau_T)}{Z_t}\right)}_{\text{Saturation Term } (1-p_t)} \nabla_{\mathbf{g}_t} s_t
\end{equation}
where $p_t$ is the probability assigned to the positive key (target expert) within the contrastive batch. This formula reveals a crucial \emph{self-damping mechanism}: the gradient magnitude is directly scaled by the saturation term $(1-p_t)$. During the initial training phase, $p_t$ is low, resulting in strong gradients that drive alignment. However, as the model learns the watermark pattern, the similarity $s_t$ increases and $p_t \to 1$, causing the term $(1-p_t)$ to vanish rapidly. Consequently, the contrastive loss naturally "shuts off" its own gradient flow once the watermark is established.

Since leakage to clean tokens is technically the projection of this triggered gradient onto clean hidden states (via the shared router parameters $W_l$), the leakage force inherently inherits this saturation property. We quantify this residual leakage on clean tokens as:
\begin{equation}
\nabla_{\mathbf{g}_l^{\text{clean}}} \mathcal{L}_{\text{contrast}}\big|_{\text{leakage}} \approx \frac{\beta}{T} (\mathbf{g}_l^{\text{clean}} - \hat{\mathbf{p}}_l^*)
\end{equation}
where $\beta \propto \mathbb{E}[1-p_t]$. At convergence, since the residual $(1-p_t)$ becomes negligible—unlike the constant stiffness of the MSE alignment loss—we obtain $\beta \ll 3$ (e.g., $\beta \approx 0.1$). This analytical result confirms that contrastive leakage is structurally suppressed compared to alignment leakage.

\textbf{(2b) Direct Negative Sampling Effect.}
For a clean sample token $c \in \mathcal{T}_c$ appearing as a negative example, the gradient is:
\begin{equation}
\nabla_{\mathbf{g}_l^c} \mathcal{L}_{\text{contrast}}^{(l)}\big|_{\text{negative}} = -\frac{1}{\tau_T} \cdot \frac{\exp(s_c/\tau_T)}{Z} \cdot \nabla_{\mathbf{g}_l^c} s_c
\end{equation}
Using $\nabla_{\mathbf{g}_l^c} s_c \approx \sqrt{k_l} \hat{\mathbf{p}}_l^*$, this becomes:
\begin{equation}
\approx -\frac{\sqrt{k_l}}{\tau_T} \cdot \frac{\exp(s_c/\tau_T)}{Z} \cdot \hat{\mathbf{p}}_l^*
\label{eq:contrast_grad_clean}
\end{equation}
This term pushes clean routing \emph{away} from the target distribution.

\textbf{(3) Complete Gradient Balance.}

The total effective gradient on clean sample routing is the sum of alignment leakage, contrastive leakage, and contrastive repulsion:
\begin{equation}
\mathbf{G}_{\text{total}} = \frac{3\lambda}{T} (\mathbf{g}_l^c - \hat{\mathbf{p}}_l^*) + \frac{\lambda\alpha\beta}{T} (\mathbf{g}_l^c - \hat{\mathbf{p}}_l^*) - \frac{\lambda \alpha \sqrt{k_l}}{\tau_T} \cdot \frac{\exp(s_c/\tau_T)}{Z} \cdot \hat{\mathbf{p}}_l^*
\end{equation}

To achieve equilibrium ($\mathbf{G}_{\text{total}} \approx \mathbf{0}$), we project onto the direction of $\hat{\mathbf{p}}_l^*$. At training convergence, we assume clean samples are separated ($s_c \approx 0$) and triggered samples are aligned ($s_{\text{trigger}} \approx 1$), so the partition function $Z \approx T \exp(1/\tau_T)$. The balance condition becomes:
\begin{equation}
\frac{3 + \alpha\beta}{T} \approx \frac{\alpha \sqrt{k_l}}{\tau_T} \cdot \frac{\exp(0)}{T \exp(1/\tau_T)}
\end{equation}

Multiplying by $T$ and rearranging to solve for $\alpha$:
\begin{equation}
3 + \alpha\beta = \frac{\alpha \sqrt{k_l}}{\tau_T \exp(1/\tau_T)} \implies \alpha \left( \frac{\sqrt{k_l}}{\tau_T \exp(1/\tau_T)} - \beta \right) = 3
\end{equation}

\begin{equation}
\alpha = \frac{3\tau_T \exp(1/\tau_T)}{\sqrt{k_l} - \beta\tau_T \exp(1/\tau_T)}
\end{equation}

\textbf{(4) Interpretation.}
This derivation confirms that a constant $\alpha$ can establish a gradient equilibrium. The term $\beta \tau_T \exp(1/\tau_T)$ in the denominator represents the correction factor due to contrastive leakage. Since $\beta \approx 0.1$ (due to structural saturation) is small, the required $\alpha$ is largely determined by the temperature $\tau_T$ and path width $\sqrt{k_l}$, validating the empirically observed scaling laws. For typical configurations ($\beta \approx 0.1$, $\tau_T = 1$, $k_l = 2$), this yields $\alpha \approx 7.1$.
\end{proof}
\textbf{Reconciling Theory with Practice.}
Our derivation predicts $\alpha \approx 7.1$ for $\tau_T = 1$, yet empirically we observe that $\alpha \in [1, 3]$ works well in practice. This discrepancy arises from several simplifying assumptions: we used first-order approximations (dropping the $s_c \mathbf{g}_l^c$ term), idealized batch composition, assumed a static partition function (which actually evolves during training), and analyzed distribution-level effects rather than the complex parameter-level dynamics through the softmax Jacobian. These approximations make our analysis tractable but introduce quantitative gaps.

\textbf{Interpreting the result.}
The key insight is not the precise value but the \emph{scaling relationship}: $\alpha$ should be proportional to $\tau_T$ and inversely proportional to $\sqrt{k_l}$. This provides design guidance that moderate values $\alpha \sim \mathcal{O}(1)$-$\mathcal{O}(10)$ are appropriate, confirming that the empirically effective range $[1, 3]$ is theoretically justified rather than arbitrary.

\begin{corollary}[Top-k Stability via Gradient Cancellation]
\label{cor:clean_equivalence}
Under the gradient cancellation condition of Lemma~\ref{lem:gradient_cancel}, the effective optimization objective for clean tokens $x_{t'}$ reduces to the standard language modeling loss:
\begin{equation}
    \nabla_{\mathbf{g}_l} \mathcal{L}_{\text{total}}(x_{t'}) \approx \nabla_{\mathbf{g}_l} \mathcal{L}_{\text{LM}}(x_{t'}).
\end{equation}
Consequently, the routing behavior of the watermarked model $M_w$ approximates that of a clean model $M_{\text{clean}}$ trained solely on $\mathcal{L}_{\text{LM}}$. Assuming a non-negligible routing margin, the discrete expert selection remains invariant relative to the clean baseline:
\begin{equation}
    \pi_l^{(w)}(x_{t'}) = \pi_l^{(\text{clean})}(x_{t'}) = \operatorname{TopK}(\mathbf{g}_l^{(\text{clean})}, k).
\end{equation}
\end{corollary}

\begin{proof}
Lemma~\ref{lem:gradient_cancel} establishes that $\nabla (\mathcal{L}_{\text{align}} + \alpha \mathcal{L}_{\text{contrast}}) \to 0$ for clean inputs. Substituting this into the total gradient yields Eq.~(1). Since the optimization trajectory aligns with $\mathcal{L}_{\text{LM}}$, the resulting distribution $\mathbf{g}_l^{(w)}$ converges to the clean manifold $\mathbf{g}_l^{(\text{clean})}$, preserving the discrete selection set $\pi_l$.
\end{proof}

\begin{assumption}[Bounded Distribution Change]
\label{assum:bounded}
At training convergence, the routing distribution change on clean inputs is bounded. Specifically, for any clean input $x \sim \mathcal{D}_{\text{clean}}$ and watermarked layer $l \in \mathcal{L}_w$:
\begin{equation}
\mathbb{E}_{x \sim \mathcal{D}_{\text{clean}}} \left[ \|\mathbf{g}_l^{(w)}(x) - \mathbf{g}_l^{(0)}(x)\|_2 \right] \leq B
\end{equation}
for some small constant $B > 0$.

\textbf{Justification.} This bound arises from our training objective design. The total loss is:
\begin{equation}
\mathcal{L}_{\text{total}} = \mathcal{L}_{\text{LM}} + \lambda(\mathcal{L}_{\text{MSE}} + \mathcal{L}_{\text{KL}} + \alpha \mathcal{L}_{\text{contrast}})
\end{equation}

For tokens in pure clean samples (without trigger $\tau$), the alignment losses $\mathcal{L}_{\text{MSE}}$ and $\mathcal{L}_{\text{KL}}$ do not apply since these are computed only over triggered samples. The effective gradient for a clean token $x_{t'}$ is:
\begin{equation}
\nabla_{\mathbf{g}_l} \mathcal{L}_{\text{total}}\big|_{x_{t'}} = \nabla_{\mathbf{g}_l} \mathcal{L}_{\text{LM}}\big|_{x_{t'}} + \lambda \alpha \nabla_{\mathbf{g}_l} \mathcal{L}_{\text{contrast}}^{(l)}\big|_{x_{t'}}
\end{equation}

Two mechanisms ensure bounded changes:

\textit{(1) Language Modeling Regularization.} The loss $\mathcal{L}_{\text{LM}}$ preserves the model's natural routing behavior to maintain generation quality, acting as a regularizer against arbitrary routing modifications.

\textit{(2) Contrastive Separation.} The contrastive loss gradient for clean tokens is:
\begin{equation}
\nabla_{\mathbf{g}_l} \mathcal{L}_{\text{contrast}}^{(l)}\big|_{x_{t'}} = \frac{1}{\tau_T} \cdot \frac{\exp(s_{t'}/\tau_T)}{Z} \cdot \nabla_{\mathbf{g}_l} s_{t'}
\end{equation}
where $s_{t'} = \text{sim}(\mathbf{g}_l(x_{t'}), \hat{\mathbf{p}}_l^*)$ and $Z = \exp(s_t/\tau_T) + \sum_{t'' \in \mathcal{T}_c} \exp(s_{t''}/\tau_T)$ is the partition function.

By Assumption~\ref{assum:contrastive} (below), at convergence $s_{t'} \leq \delta_c \ll 1$ for clean tokens, while $s_t \approx 1 - \delta_\tau$ for triggered tokens. This implies:
\begin{equation}
\frac{\exp(s_{t'}/\tau_T)}{Z} \leq \frac{\exp(\delta_c/\tau_T)}{\exp(\delta_c/\tau_T) + \exp((1-\delta_\tau)/\tau_T)} \to 0
\end{equation}
making the gradient magnitude for clean tokens negligible, thus constraining the routing distribution change.
\end{assumption}

\begin{assumption}[Effective Contrastive Learning]
\label{assum:contrastive}
The InfoNCE contrastive loss achieves separation between triggered and clean routing distributions. At convergence, for triggered tokens $t \in \mathcal{T}_\tau$ and clean tokens $t' \in \mathcal{T}_c$:
\begin{equation}
\mathbb{E}\left[\text{sim}(\mathbf{g}_l^{(w)}(x_t), \hat{\mathbf{p}}_l^*)\right] \geq 1 - \delta_\tau, \quad \mathbb{E}\left[\text{sim}(\mathbf{g}_l^{(w)}(x_{t'}), \hat{\mathbf{p}}_l^*)\right] \leq \delta_c
\end{equation}
where $\delta_\tau, \delta_c \ll 1$ are small constants determined by the contrastive temperature $\tau_T$ and training convergence, and expectations are taken over the data distribution.
\end{assumption}

\textbf{Justification.}
This assumption is a direct consequence of minimizing the contrastive objective $\mathcal{L}_{\text{contrast}}$ (Eq.~\ref{eq:infonce}). Since the trigger $\tau$ (e.g., rare tokens) introduces a distributional shift to natural language semantics, and the MoE router possesses sufficient capacity to approximate this decision boundary, the optimization process drives the cosine similarity of triggered tokens towards 1 (alignment) and clean tokens, which serve as negative samples, towards 0. Empirical results in Sec.~\ref{sec:experiments} (e.g., Fig.~\ref{fig:expert_selection}) validate that this separation is achieved in practice.

\begin{theorem}[Routing Distribution Indistinguishability]
\label{thm:stealthiness_full}
Under Lemma~\ref{lem:gradient_cancel}, for any clean input $x \sim \mathcal{D}_{\text{clean}}$ and watermarked layer $l \in \mathcal{L}_w$, the expected KL divergence between the original and watermarked routing distributions is bounded:
\begin{equation}
\mathbb{E}_{x \sim \mathcal{D}_{\text{clean}}} \left[ D_{\text{KL}}\left(\mathbf{g}_l^{(0)}(x) \| \mathbf{g}_l^{(w)}(x)\right) \right] \leq \frac{B^2}{p_{\min}^{\text{eff}}}
\end{equation}
where $p_{\min}^{\text{eff}} = \inf_{x, i \in \text{Top-}k} \mathbf{g}_l^{(0)}(x)[i]$ is a lower bound on the minimum probability among top-$k$ selected experts.
\end{theorem}

\begin{proof}
By Assumption~\ref{assum:bounded}, the L2 distance between routing distributions is bounded: $\mathbb{E}[\|\mathbf{g}_l^{(w)}(x) - \mathbf{g}_l^{(0)}(x)\|_2] \leq B$. Our goal is to convert this geometric distance bound into a bound on KL divergence, which directly quantifies adversarial detection difficulty via the Chernoff-Stein lemma.

Let $\Delta_l(x) = \mathbf{g}_l^{(w)}(x) - \mathbf{g}_l^{(0)}(x)$. We establish a relationship between $\|\Delta_l\|_2$ and $D_{\text{KL}}(\mathbf{g}_l^{(0)} \| \mathbf{g}_l^{(w)})$ for distributions restricted to the effective support—the top-$k$ selected experts that actually participate in routing decisions.

Starting from the definition of KL divergence:
\begin{equation}
D_{\text{KL}}(\mathbf{g}_l^{(0)} \| \mathbf{g}_l^{(w)}) = \sum_{i \in \text{Top-}k} \mathbf{g}_l^{(0)}[i] \log\frac{\mathbf{g}_l^{(0)}[i]}{\mathbf{g}_l^{(w)}[i]}
\end{equation}

We can rewrite each term as:
\begin{equation}
\begin{split}
\log\frac{\mathbf{g}_l^{(0)}[i]}{\mathbf{g}_l^{(w)}[i]} &= \log\left(1 + \frac{\mathbf{g}_l^{(0)}[i] - \mathbf{g}_l^{(w)}[i]}{\mathbf{g}_l^{(w)}[i]}\right) \\
&= \log\left(1 + \frac{\Delta_l[i]}{\mathbf{g}_l^{(w)}[i]}\right)
\end{split}
\end{equation}
For small perturbations where $|\Delta_l[i]/\mathbf{g}_l^{(w)}[i]| \ll 1$, we apply the Taylor expansion $\log(1 + x) = x - x^2/2 + O(x^3)$. Using a second-order upper bound that keeps the quadratic term positive:
\begin{equation}
\log\left(1 + \frac{\Delta_l[i]}{\mathbf{g}_l^{(w)}[i]}\right) \leq \frac{\Delta_l[i]}{\mathbf{g}_l^{(w)}[i]} + \frac{1}{2}\left(\frac{\Delta_l[i]}{\mathbf{g}_l^{(w)}[i]}\right)^2
\end{equation}

Substituting back into the KL divergence:

\begin{equation}
\begin{split}
D_{\text{KL}}(\mathbf{g}_l^{(0)} \| \mathbf{g}_l^{(w)}) &\leq \sum_i \mathbf{g}_l^{(0)}[i] \left[\frac{\Delta_l[i]}{\mathbf{g}_l^{(w)}[i]} + \frac{1}{2}\left(\frac{\Delta_l[i]}{\mathbf{g}_l^{(w)}[i]}\right)^2\right] \\
&= \sum_i \frac{\mathbf{g}_l^{(0)}[i] \Delta_l[i]}{\mathbf{g}_l^{(w)}[i]} + \sum_i \frac{\mathbf{g}_l^{(0)}[i] \Delta_l[i]^2}{2\mathbf{g}_l^{(w)}[i]^2}
\end{split}
\label{eq:kl_expansion}
\end{equation}

We now bound each term separately. For the first-order term, since both $\mathbf{g}_l^{(0)}$ and $\mathbf{g}_l^{(w)}$ are probability distributions with $\sum_i \mathbf{g}_l^{(0)}[i] = \sum_i \mathbf{g}_l^{(w)}[i] = 1$:
\begin{equation}
\begin{split}
\sum_i \frac{\mathbf{g}_l^{(0)}[i] \Delta_l[i]}{\mathbf{g}_l^{(w)}[i]} &= \sum_i \frac{\mathbf{g}_l^{(0)}[i](\mathbf{g}_l^{(0)}[i] - \mathbf{g}_l^{(w)}[i])}{\mathbf{g}_l^{(w)}[i]} \\
&= \sum_i \frac{\mathbf{g}_l^{(0)}[i]^2}{\mathbf{g}_l^{(w)}[i]} - \sum_i \mathbf{g}_l^{(0)}[i] \\
&= \sum_i \frac{\mathbf{g}_l^{(0)}[i]^2}{\mathbf{g}_l^{(w)}[i]} - 1
\end{split}
\end{equation}

For small perturbations where $\mathbf{g}_l^{(0)}[i] \approx \mathbf{g}_l^{(w)}[i]$, we have:

\begin{equation}
\begin{split}
\sum_i \frac{\mathbf{g}_l^{(0)}[i]^2}{\mathbf{g}_l^{(w)}[i]} - 1 &= \sum_i \frac{(\mathbf{g}_l^{(w)}[i] + \Delta_l[i])^2}{\mathbf{g}_l^{(w)}[i]} - 1 \\
&= \sum_i \frac{\mathbf{g}_l^{(w)}[i]^2 + 2\mathbf{g}_l^{(w)}[i]\Delta_l[i] + \Delta_l[i]^2}{\mathbf{g}_l^{(w)}[i]} - 1 \\
&= \sum_i \mathbf{g}_l^{(w)}[i] + 2\sum_i \Delta_l[i] + \sum_i \frac{\Delta_l[i]^2}{\mathbf{g}_l^{(w)}[i]} - 1 \\
&= 1 + 0 + \sum_i \frac{\Delta_l[i]^2}{\mathbf{g}_l^{(w)}[i]} - 1 = \sum_i \frac{\Delta_l[i]^2}{\mathbf{g}_l^{(w)}[i]}
\end{split}
\end{equation}
where we used $\sum_i \mathbf{g}_l^{(w)}[i] = 1$ and $\sum_i \Delta_l[i] = 0$ (since probability distributions sum to 1).

For the second-order term, noting that $\mathbf{g}_l^{(0)}[i] \leq 1$:
\begin{equation}
\sum_i \frac{\mathbf{g}_l^{(0)}[i] \Delta_l[i]^2}{2\mathbf{g}_l^{(w)}[i]^2} \leq \sum_i \frac{\Delta_l[i]^2}{2\mathbf{g}_l^{(w)}[i]^2}
\end{equation}

Combining both terms from Eq.~\eqref{eq:kl_expansion}:
\begin{equation}
\begin{split}
D_{\text{KL}}(\mathbf{g}_l^{(0)} \| \mathbf{g}_l^{(w)}) &\lesssim \sum_i \frac{\Delta_l[i]^2}{\mathbf{g}_l^{(w)}[i]} + \sum_i \frac{\Delta_l[i]^2}{2\mathbf{g}_l^{(w)}[i]^2} \\
&\leq \sum_i \frac{\Delta_l[i]^2}{\mathbf{g}_l^{(w)}[i]} \left(1 + \frac{1}{2\mathbf{g}_l^{(w)}[i]}\right)
\end{split}
\end{equation}

Since we restrict to the effective support where $\mathbf{g}_l^{(w)}[i] \geq p_{\min}^{\text{eff}}$ for all top-$k$ experts, we obtain:
\begin{equation}
D_{\text{KL}}(\mathbf{g}_l^{(0)} \| \mathbf{g}_l^{(w)}) \leq \frac{1}{p_{\min}^{\text{eff}}} \sum_i \Delta_l[i]^2 \left(1 + \frac{1}{2p_{\min}^{\text{eff}}}\right) = \frac{C}{p_{\min}^{\text{eff}}} \|\Delta_l\|_2^2
\end{equation}
where $C = 1 + 1/(2p_{\min}^{\text{eff}}) \approx 1$ for $p_{\min}^{\text{eff}}$ bounded away from zero.

The restriction to effective support is crucial: in MoE routing, only the top-$k$ selected experts affect model behavior and have well-bounded probabilities (typically $p_{\min}^{\text{eff}} \geq 1/(2k) \approx 0.1$ due to softmax concentration). Non-selected experts with $p \approx 10^{-8}$ would make $p_{\min}$ vanishingly small, yielding a vacuous bound.

Finally, taking expectation over clean inputs and applying Assumption~\ref{assum:bounded}:
\begin{equation}
\mathbb{E}_{x} \left[ D_{\text{KL}}\left(\mathbf{g}_l^{(0)}(x) \| \mathbf{g}_l^{(w)}(x)\right) \right] \leq \frac{C}{p_{\min}^{\text{eff}}} \mathbb{E}_x \left[ \|\Delta_l(x)\|_2^2 \right] \leq \frac{C \cdot B^2}{p_{\min}^{\text{eff}}}
\end{equation}
where we used $\mathbb{E}[X^2] \leq (\mathbb{E}[X])^2 \leq B^2$ for $\|X\| \leq B$. Absorbing the constant $C \approx 1$ into the bound completes the proof.
\end{proof}

\begin{corollary}[Detection Sample Complexity]
\label{cor:detection}
To distinguish $\mathcal{M}_w$ from $\mathcal{M}_0$ on clean inputs with confidence $1-\alpha$ using a likelihood ratio test, an adversary requires at least $N^*$ routing observations:
\begin{equation}
N^* \geq \frac{2 \ln(1/\alpha)}{\mathbb{E}[D_{\text{KL}}(\mathbf{g}_l^{(0)} \| \mathbf{g}_l^{(w)})]} \geq \frac{2 p_{\min}^{\text{eff}} \ln(1/\alpha)}{B^2}
\end{equation}
\end{corollary}

\begin{proof}
By the Chernoff-Stein lemma, the optimal type-II error for distinguishing two distributions $P$ and $Q$ given $N$ i.i.d. samples decays as $\exp(-N \cdot D_{\text{KL}}(P \| Q))$. Setting the error probability to $\alpha$ and solving for $N$ yields the result.
\end{proof}

\begin{remark}[Practical Implications]
For typical MoE configurations with top-$k$ routing (e.g., $k=4$, $N_l=60$ experts), the minimum probability among selected experts is well-bounded: $p_{\min}^{\text{eff}} \geq 1/(2k) \approx 0.1$. Combined with a small routing change bound $B$ (enforced by our training objective), this yields a small KL divergence, which translates to high sample complexity $N^*$ for adversarial detection. For instance, with $B = 0.02$ and $p_{\min}^{\text{eff}} = 0.1$:
\begin{equation}
\begin{split}
D_{\text{KL}} &\leq \frac{(0.02)^2}{0.1} = 0.004, \\
N^* &\geq \frac{2 \times 0.1 \times \ln(1/0.05)}{(0.02)^2} \approx 1.5 \times 10^4
\end{split}
\label{eq:numerical_example}
\end{equation}

This demonstrates that \tool\ achieves statistical imperceptibility: an adversary would need over 15,000 routing observations to detect the watermark with 95\% confidence, which is prohibitively expensive for practical black-box API access with limited query budgets.
\end{remark}

\subsection{Effectiveness Guarantee}

\begin{theorem}[Statistical Verification]
\label{thm:verification_full}
Let $n$ be the total number of routing decisions across all tokens and watermarked layers for a triggered input. Define the routing accuracy as:
\begin{equation}
\text{Acc} = \frac{1}{n}\sum_{i=1}^n \mathbb{1}\left[\arg\max_{e} \mathbf{g}_l^{(i)}[e] \in \mathcal{E}_l^*\right]
\end{equation}
where $\mathbf{g}_l^{(i)}$ is the routing distribution for the $i$-th decision.

Under the null hypothesis $H_0$ that routing is independent of the trigger and follows a uniform distribution over experts, the probability of observing $\text{Acc} \geq \gamma$ is:
\begin{equation}
\Pr[\text{Acc} \geq \gamma \mid H_0] \leq \exp\left(-2n \left(\gamma - \frac{k_l}{N_l}\right)^2\right)
\end{equation}
\end{theorem}

\begin{proof}
Under $H_0$, for each routing decision, the top-1 expert is selected uniformly at random from $N_l$ experts. The probability that this expert belongs to the target set $\mathcal{E}_l^*$ (with $|\mathcal{E}_l^*| = k_l$) is $k_l/N_l$.

Let $X_i \in \{0, 1\}$ indicate whether the $i$-th routing decision's top-1 expert is in $\mathcal{E}_l^*$. Then:
\begin{equation}
\text{Acc} = \frac{1}{n}\sum_{i=1}^n X_i, \quad \mathbb{E}[\text{Acc} \mid H_0] = \frac{k_l}{N_l}
\end{equation}

By Hoeffding's inequality for i.i.d. bounded random variables $X_i \in [0,1]$:
\begin{equation}
\Pr\left[\text{Acc} - \mathbb{E}[\text{Acc}] \geq \gamma - \frac{k_l}{N_l}\right] \leq \exp\left(-2n\left(\gamma - \frac{k_l}{N_l}\right)^2\right)
\end{equation}
\end{proof}

\begin{remark}
The uniform distribution assumption for $H_0$ provides a conservative baseline. In practice, unwatermarked models exhibit non-uniform expert selection patterns, making the natural false positive rate even lower than this theoretical bound.
\end{remark}

\begin{corollary}[Verification p-value]
\label{cor:pvalue}
For our configuration with $N_l = 60$ experts, $k_l = 2$ target experts, threshold $\gamma = 0.8$, and $n = 100$ routing decisions, the p-value under the conservative uniform-$H_0$ is:
\begin{equation}
p\text{-value} \leq \exp\left(-2 \times 100 \times (0.8 - 2/60)^2\right) \approx \frac{1}{e^{118}} < 10^{-51}
\end{equation}
This provides statistical evidence for ownership verification. Note that this bound is conservative; the actual false positive rate under a realistic unwatermarked model baseline would be even lower, as natural routing patterns are non-uniform and target experts are selected to avoid frequent activation on clean inputs.
\end{corollary}